\documentclass[10pt,twocolumn,twoside] {IEEEtran}
\pdfoutput=1
\usepackage{amsmath,amsfonts}
\usepackage{algorithmic}
\usepackage{array}
\usepackage[caption=false,font=footnotesize,labelfont=sf,textfont=sf]{subfig}
\usepackage{textcomp}
\usepackage{stfloats}
\usepackage{url}
\usepackage{verbatim}
\usepackage{graphicx}
\usepackage{float}
\usepackage{color}
\usepackage{amssymb}
\usepackage{makecell}

\def\BibTeX{{\rm B\kern-.05em{\sc i\kern-.025em b}\kern-.08em
    T\kern-.1667em\lower.7ex\hbox{E}\kern-.125emX}}
\usepackage{balance}

\usepackage[dvipsnames]{xcolor}
\usepackage{cite}

\newtheorem{thm}{Theorem}

\newenvironment{proof}{{\indent \indent \it Proof:}}{\hfill $\blacksquare$\par}
\usepackage{algorithm}
\usepackage{algorithmic}

\usepackage{xpatch}
\makeatletter
\ExplSyntaxOn
\cs_new:Npn \bibColoredItems #1#2
  {
    \clist_map_inline:nn {#2} { \cs_new:cpn {bib@colored@##1} {#1} } 
  }
\ExplSyntaxOff
\newcommand\bib@setcolor[1]{
  \ifcsname bib@colored@#1\endcsname
    \expanded{\noexpand\color{\csname bib@colored@#1\endcsname}}
  \else
    \normalcolor
  \fi
}

\makeatother

\begin{document}
\title{Optimizing Clustered Cell-Free Networking for Sum Ergodic Capacity Maximization with Joint Processing Constraint}

\author{Funing~Xia, \emph{Student Member,~IEEE}, Junyuan~Wang, \emph{Member,~IEEE}, and Lin~Dai, \emph{Senior Member,~IEEE}
\thanks{Received 19 March 2024; revised 29 August 2024; accepted 7 November
2024.
This work was supported in part by the National Natural Science
Foundation of China under Grant 62371344 and Grant 62001330 and in
part by the Fundamental Research Funds for Central Universities.
The work of Lin Dai was supported by Huawei Technologies Company Ltd., under Grant 9229002.
An earlier version of this paper was presented in part at the IEEE International Symposium on Information Theory (ISIT), Taipei, Taiwan, China in June 2023 \cite{xfn_ISIT2023}.
The associate editor coordinating the review of this article and approving it for publication was H. Zhu.
\textit{(Corresponding author: Junyuan Wang.)}}
\thanks{F. Xia is with the College of Electronic and Information Engineering, Tongji University, Shanghai, 201804, China (e-mail: panzerxia@tongji.edu.cn).}
\thanks{J. Wang is with the College of Electronic and Information Engineering, the Institute of Advanced Study, and Shanghai Institute of Intelligent Science and Technology, Tongji University, Shanghai, 201804, China (e-mail: junyuanwang@tongji.edu.cn).}
\thanks{L. Dai is with the Department of Electrical Engineering, City University of Hong Kong, Hong Kong SAR, China (email: lindai@cityu.edu.hk).}
}


\maketitle

\begin{abstract}
Clustered cell-free networking has been considered as an effective scheme to trade off between the low complexity of current cellular networks and the superior performance of fully cooperative networks.
With clustered cell-free networking, the wireless network is decomposed into a number of disjoint parallel operating subnetworks with joint processing adopted inside each subnetwork independently for intra-subnetwork interference mitigation.
Different from the existing works that aim to maximize the number of subnetworks without considering the limited processing capability of base-stations (BSs), this paper investigates the clustered cell-free networking problem with the objective of maximizing the sum ergodic capacity while imposing a limit on the number of user equipments (UEs) in each subnetwork to constrain the joint processing complexity.
By successfully transforming the combinatorial NP-hard clustered cell-free networking problem into an integer convex programming problem, the problem is solved by the branch-and-bound method.
To further reduce the computational complexity, a bisection clustered cell-free networking ($\text{B}\text{C}^2\text{F}$-Net) algorithm is proposed to decompose the network hierarchically.
Simulation results show that compared to the branch-and-bound based scheme, the proposed $\text{B}\text{C}^2\text{F}$-Net algorithm significantly reduces the computational complexity yet achieves nearly the same network decomposition result.
{Moreover, our $\text{B}\text{C}^2\text{F}$-Net algorithm achieves near-optimal performance and outperforms the state-of-the-art benchmarks with up to 25\% capacity gain.}
\end{abstract}


\section{Introduction}
\label{Introduction}
Cellular network architecture has been used since the first generation mobile communication system, with which the whole network is decomposed into a set of cells.
In each cell, a base station (BS) is located at its center and provides services to the user equipments (UEs) in it.
Since the cells operate independently, the UEs located in the cell-edge areas are subject to strong inter-cell interference with significant performance degradation, which is popularly known as the cell-edge problem in cellular networks.
With the BSs being densely deployed to provide high-data-rate services and seamless coverage, a large number of UEs would fall into the cell-edge areas and thus suffer from the cell-edge problem.
This is becoming a bottleneck of boosting the rate performance of future mobile communication systems via densifying the BS deployment \cite{improving_dense_network_performance, optimal_decomposition_networks, clustered_cell_free_networking, rate_constrained_decomposition,  c2_what_should_future_network_be,cgn,zoy,das_JW}.

To circumvent the cell-edge problem rooted in cellular networks, joint processing among BSs has been considered in the literature, which has been applied in various promising technologies such as distributed antenna system (DAS) \cite{ul_capacity_study_colocated_distributed_antennas, asymptotic_rate_analysis_dl_das_JW, distributed_vs_microcellular}, coordinated multi-point (CoMP) transmission  \cite{network_coordination_comp}, \cite{comp_uc_adaptive_clustering}, network multiple-input-multiple-output (MIMO) \cite{network_mimo_origin, dynamic_coalition_network_mimo}, cloud radio access network (C-RAN) \cite{c_ran_toward_green_ran, c_ran_overview,pa_robust_tx_uc_c_ran} and the recently popular cell-free massive MIMO \cite{cf_vs_small_cell,user_centric_cf_mmimo_survey}.
It has been demonstrated that coordinating all BSs for joint processing outperforms the traditional cellular network in terms of both per-user rate and spectral efficiency \cite{ul_capacity_study_colocated_distributed_antennas, asymptotic_rate_analysis_dl_das_JW, distributed_vs_microcellular, cf_vs_small_cell}. 
Since all the BSs jointly serve the UEs, there are no cells any more, and thus no cell-edge problem.
Evolving the mobile communication system from cellular network architecture to cell-free architecture is a promising solution to the cell-edge problem, which has attracted considerable interests from both academia and industry.

Despite its superior performance, the aforementioned fully cooperative cell-free network needs a central processing unit (CPU) to collect the channel state information (CSI) between all UEs and BSs and perform centralized joint signal processing for both downlink data transmission and uplink signal detection.
As a result, the signaling overhead and joint processing complexity would sharply increase with the number of UEs and BSs in the network, making the fully cooperative cell-free network architecture unscalable \cite{das_JW, joint_power_antenna_selection_c_ran}.
Moreover, serving UEs by distant BSs is inefficient as the distant BSs contribute little to the received signal power while occupying valuable spectral and energy resources.
Motivated by this, a number of research works considered associating each UE with multiple BSs and jointly optimized UE-BS association and resource allocation \cite{stable_matching_urllc, ee_ua_pa_multi_conn_mmwave_net,multi_conn_ua_pa_mmwave_net,ee_ua_iiot} for interference management.
However, such joint optimization requires huge channel state information exchange and high computational complexity.
Other research works were devoted to a simple UE-centric virtual-cell based network architecture, with which each UE associates with a number of neighboring BSs as its serving virtual cell \cite{ul_capacity_analysis_das, das_JW, uc_jt_virtual_cell_udn, uc_5g_cellular_network}.
Since a UE is served only by the BSs in its virtual cell, the cell-edge problem can be completely avoided and the joint processing complexity remains the same as the number of BSs and UEs grows.
Yet different UEs’ virtual cells interfere with each other, and the interference could be extremely strong when they are largely overlapping, i.e., sharing a large portion of BSs \cite{das_JW,min_separation_clustering_udn}.

To keep the joint processing complexity at an acceptable level and meanwhile minimize the inter-subnetwork interference, the \textit{clustered cell-free networking} concept has been proposed recently \cite{optimal_decomposition_networks,clustered_cell_free_networking,cgn,c2_what_should_future_network_be,zoy,rate_constrained_decomposition}.
The idea is to dynamically partition the network into non-overlapping subnetworks by clustering strongly interfering UEs and BSs into the same subnetwork.\footnote{Here, two types of nodes, BSs and UEs, are clustered into subnetworks with clustered cell-free networking. Similar problems that involve only one type of nodes in other networks such as internet of things (IoT) network, Ad-Hoc network, device-to-device (D2D) network and sensor network can be found in \cite{rance_clustering_adhoc, thwsn_clustering_sensor_network, clustered_d2d_net}.}
As such, joint processing only needs to be performed within each subnetwork to cancel the intra-subnetwork interference and the subnetworks can operate independently.
Since a BS provides services only to the UEs in the same subnetwork rather than all the UEs, the joint processing complexity can be greatly reduced and the energy efficiency is improved compared to the fully cooperative network \cite{clustered_cell_free_networking,rate_constrained_decomposition}.
The key to exploiting the advantages of clustered cell-free networking to the full is to partition the network into subnetworks optimally.
The existing algorithms can be classified as two-stage or one-stage algorithms according to whether the BSs and UEs are clustered into subnetworks separately or not.

\subsection{Two-Stage Clustered Cell-Free Networking}
For two-stage clustered cell-free networking algorithms, either BSs or UEs are first clustered into subnetworks and then the remaining devices are assigned into the subnetworks according to their BS-UE association preference.
One way is to merge cellular cells into subnetworks (equivalent to clustering BSs), which is known as CoMP and can be regarded as \textit{BS-centric} two-stage clustered cell-free networking.
Various CoMP clustering schemes were proposed in  \cite{static_clustering_comp, dynamic_cell_clustering_design, dynamic_rr_clustering_user_scheduling}, where the set of potential cooperating BS clusters was usually predefined and the CoMP clusters were selected among them based on different types of performance metrics.
Dynamic BS-centric clustering was investigated in  \cite{dynamic_clustering_multicell_cooperative_processing, dynamic_clustering_mu_das, virtual_cell_clustering_ra} to achieve better system performance.
Specifically, by fixing the number of cooperating BSs in each cluster at a constant, a greedy algorithm was proposed in \cite{dynamic_clustering_multicell_cooperative_processing} to dynamically group BSs into joint processing clusters for sum rate maximization.
In \cite{dynamic_clustering_mu_das}, to improve the flexibility, the maximum number of BSs in each subnetwork is constrained and the BSs are clustered based on a predefined signal-interference matrix.
Without placing any limit on the BS cluster size in \cite{virtual_cell_clustering_ra}, BSs are first clustered via hierarchical clustering based on the minimax linkage criterion and then UEs are assigned to the BS clusters according to either the Euclidean distance or the channel quality.
As the above BS-centric approaches form subnetworks from the BS side only, there would still exist UEs located at the boundaries of the subnetworks, suffering from strong inter-subnetwork interference \cite{optimal_decomposition_networks}.

To completely avoid the edge problem, \textit{UE-centric} two-stage clustering algorithms were proposed, where the subnetworks are generated via grouping UEs \cite{das_JW, ee_large_scale_das,cf_mmimo_joint_uc_aps, zoy}. 
In \cite{das_JW, ee_large_scale_das}, each UE first chooses its own serving BSs to form the so-called virtual cell, and then virtual cells are merged (equivalent to clustering UEs) to generate subnetworks.
Specifically, the virtual cells sharing at least one BS were merged in \cite{das_JW}, while in \cite{ee_large_scale_das}, the virtual cells causing strong mutual interference were merged.
As UEs are always surrounded by their serving BSs, there are no subnetwork-edge UEs any more.
However, experimental results showed that the number of subnetworks after merging virtual cells varies with different network layouts and there usually exists one giant subnetwork \cite{das_JW, optimal_decomposition_networks, clustered_cell_free_networking, rate_constrained_decomposition}, still leading to high joint processing complexity.
By contrast, in \cite{cf_mmimo_joint_uc_aps, zoy}, hierarchical clustering algorithms were applied to first group UEs into a predefined number of clusters based on the similarity of their channels, and then assigned BSs to UE clusters according to channel quality \cite{cf_mmimo_joint_uc_aps} or following the nearest neighbor principle \cite{zoy}.
Similarly, a cosine distance based similarity metric was developed in \cite{clustered_cf_mmimo}, and K-means clustering algorithm was then iteratively applied to group UEs into clusters with the restriction that the number of UEs in each cluster is no larger than the number of orthogonal pilot sequences for the sake of avoiding pilot contamination in each subnetwork.

\subsection{One-Stage Optimal Clustered Cell-Free Networking}
Although the aforementioned two-stage clustered cell-free networking schemes, either BS-centric or UE-centric, are easy to be implemented, they can only lead to suboptimal network decomposition and provide no performance guarantee.
In fact, it has been shown in \cite{optimal_decomposition_networks} that the optimal clustered cell-free networking should be determined based on the information from both the BS and UE sides.
A few recent works \cite{optimal_decomposition_networks,rate_constrained_decomposition, clustered_cell_free_networking,c2_what_should_future_network_be,cgn} researched the clustered cell-free networking problem by partitioning the network into subnetworks in one stage to optimize certain performance metrics.
Since the average subnetwork size is smaller when there are more subnetworks, most of them aimed to maximize the number of subnetworks for the purpose of reducing the joint processing complexity and signaling overhead as much as possible \cite{optimal_decomposition_networks,rate_constrained_decomposition, clustered_cell_free_networking,c2_what_should_future_network_be}. 
Specifically, the clustered cell-free networking problem was formulated in \cite{optimal_decomposition_networks} as a graph partitioning problem by modeling the network as a weighted undirected bipartite graph and a binary search based spectral relaxation (BSSR) algorithm was proposed to solve it.
The authors of \cite{rate_constrained_decomposition, clustered_cell_free_networking} further considered the downlink per-UE data-rate requirement for service quality guarantee and the sleep mode operation of BSs for energy saving.
By redefining the edge weight of the bipartite graph representing the network, a new rate-constrained network decomposition (RC-NetDecomp) algorithm was proposed in \cite{rate_constrained_decomposition, clustered_cell_free_networking}.
For the uplink transmission, the number of subnetworks was maximized in \cite{c2_what_should_future_network_be} while ensuring that the uplink sum capacity of each subnetwork is larger than some given threshold.
 
However, the number of subnetworks fails in measuring the joint processing complexity when the subnetworks are imbalanced with different sizes.
In \cite{cgn}, the minimum capacity of subnetworks was maximized by assuming that the number of subnetworks is given.
In fact, the number of subnetworks and the corresponding network decomposition should be jointly optimized to maximize the system performance while constraining the joint processing complexity of each subnetwork according to the practical requirement.
How to properly model the joint processing complexity and further optimize the clustered cell-free networking under the joint processing constraint is of paramount importance for exploiting the full potentials of future ultra-dense wireless networks yet very challenging, which still remains largely unknown. 

\subsection{Our Contributions}
In this paper, we focus on the clustered cell-free networking problem for a large-scale wireless network.
Specifically, we aim to maximize the sum ergodic capacity by jointly optimizing the number of subnetworks and the corresponding network decomposition, under the joint processing constraint that the number of UEs in each subnetwork is bounded by a preset limit.
Such a joint optimization problem is combinatorial and NP-hard to solve.
The main contributions of this paper are summarized as follows:
\begin{itemize}
\item To decouple the joint optimization of the number of subnetworks and the network decomposition, we first investigate the impact of the number of subnetworks on the sum ergodic capacity and theoretically obtain the optimal number of subnetworks in closed form, which simplifies the clustered cell-free networking problem. 
\item With the optimal number of subnetworks, we then successfully transform the intractable combinatorial NP-hard clustered cell-free networking problem into a convex integer programming problem that can be optimally solved by the branch-and-bound approach.
Simulation results show that the branch-and-bound based scheme can produce balanced subnetworks with the number of UEs in each subnetwork successfully capped by the given limit.
\item To reduce the computational complexity of the branch-and-bound based scheme, we further propose a bisection clustered cell-free networking ($\text{B}\text{C}^2\text{F}$-Net) algorithm to decompose the network hierarchically.
Simulation results corroborate that compared to the original branch-and-bound based scheme, our $\text{B}\text{C}^2\text{F}$-Net algorithm leads to much lower complexity but with slight degradation in sum capacity.
Specifically, the running time could be reduced to as low as 0.06\% of that with the branch-and-bound method with at most 4.2\% degradation in sum capacity.
\item We compare the proposed $\text{B}\text{C}^2\text{F}$-Net algorithm with a number of state-of-the-art benchmarking schemes, including a UE-centric two-stage algorithm \cite{clustered_cf_mmimo}, a BS-centric two-stage algorithm \cite{clustered_cf_mmimo} and the one-stage BSSR based algorithm \cite{optimal_decomposition_networks}.
Simulation results demonstrate that our $\text{B}\text{C}^2\text{F}$-Net algorithm achieves up to {25}\% higher average sum ergodic capacity than the benchmarks.
\end{itemize}

The remainder of this paper is organized as follows.
Section \ref{System_Model_and_Problem_Formulation} introduces the system model and formulates the clustered cell-free networking problem.
Section \ref{Proposed_Solution} transforms the optimization problem into a convex integer optimization problem and presents the branch-and-bound based solution.
A bisection clustered cell-free networking algorithm with lower complexity is proposed in Section \ref{Suboptimal_Low_Complexity_Algorithm}. 
Simulation results are presented in Section \ref{Simulation_Results}.
Finally, Section \ref{Conclusion} concludes the paper.

Throughout this paper, scalars are denoted in lower-case and upper-case letters, while their bold counterparts represent vectors and matrices, respectively.
$\mathbb{E}\{\cdot\}$, $(\cdot)^T$ and $(\cdot)^H$ denote the expectation, transpose and conjugate transpose, respectively.
$\lceil\cdot\rceil$ and $\lfloor\cdot\rfloor$ are the ceiling and floor operators.
$|\mathcal{X}|$ denotes the cardinality of set $\mathcal{X}$.
$\mathcal{G}{=}(\mathcal{V}, \mathcal{E})$ denotes a graph $\mathcal{G}$ with vertex set $\mathcal{V}$ and edge set $\mathcal{E}$.
$\mathcal{G}[\mathcal{X}]$ denotes the induced subgraph from graph $\mathcal{G}$ with vertex subset $\mathcal{X}$.
$\textbf{I}_{N}$ denotes an $N \times N$ identity matrix, and $diag\{\cdot\}$ represents a diagonal matrix.

\section{System Model and Problem Formulation}
\label{System_Model_and_Problem_Formulation}
\subsection{System Model}
\label{System_Model}
\begin{figure}[!t]
\centering
\includegraphics[width=0.4\textwidth,height=2.2in]{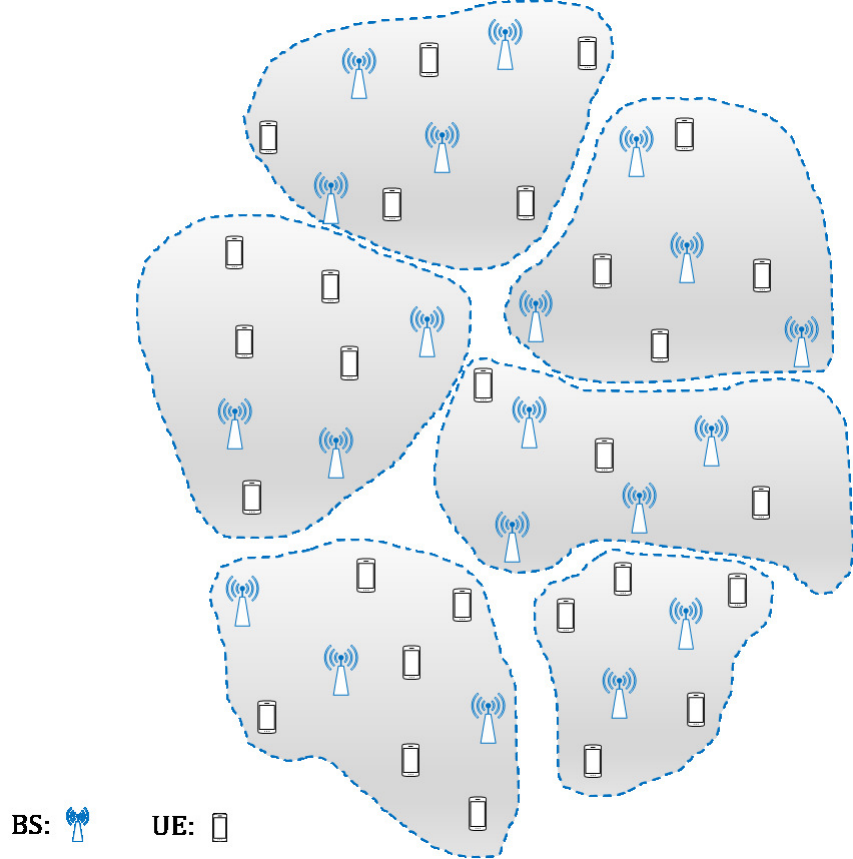}
\caption{Illustration of clustered cell-free network.}
\label{clustered_cf_network_architecture}
\end{figure}

Consider a large-scale wireless network containing $K$ single-antenna UEs and $L$ single-antenna BSs.
The set of UEs and the set of BSs are represented as $\mathcal{U}=\{u_1,u_2,\cdots,u_K\}$ and $\mathcal{B}=\{b_1,b_2,\cdots,b_L\}$, respectively, with $\vert\mathcal{U}\vert=K$ and $\vert\mathcal{B}\vert=L$.
Similar to \cite{optimal_decomposition_networks, rate_constrained_decomposition, clustered_cell_free_networking}, the wireless network can be modeled as a weighted undirected bipartite graph $\mathcal{G}=(\mathcal{V}, \mathcal{E})$, where $\mathcal{V}=\mathcal{U} \cup \mathcal{B}$ is the vertex set including all the UEs and BSs, and $\mathcal{E}=\{(u_k, b_l)|u_k \in \mathcal{U}, b_l\in\mathcal{B}\}$ is the edge set representing the channels between UEs and BSs.
Particularly, the vertex set $\mathcal{V}$ is rewritten as $\mathcal{V}=\{v_1,v_2,\cdots,v_{K+L}\}$ with
\begin{align}
v_i=
\left\{
\begin{array}{cl}
u_i, &\text{if } 1 \leq i \leq K,\\
b_{i-K}, &\text{if } K+1 \leq i \leq K+L.
\end{array}
\right.
\end{align}
The bipartite graph $\mathcal{G}=(\mathcal{V}, \mathcal{E})$ can be described by a weighted adjacency matrix $\textbf{A} \in \mathbb{R}^{(K+L)\times(K+L)}$,
given by
\begin{align}
\label{adjacency_matrix_definition}
\textbf{A}=
\begin{bmatrix}
\textbf{0}_{K{\times}K} & \textbf{W} \\
\textbf{W}^T & \textbf{0}_{L{\times}L} \\
\end{bmatrix},
\end{align}
where $\textbf{W} \in \mathbb{R}^{K \times L}$ is the weight matrix of the edges, whose $k$-th row and $l$-th column entry $w_{kl}$ is the edge weight between UE $u_k$ and BS $b_l$.

With clustered cell-free networking, as illustrated in Fig. \ref{clustered_cf_network_architecture}, the network is decomposed into multiple  disjoint subnetworks and each subnetwork operates independently.

Assume that there are $M$ subnetworks in total and the set of UEs and BSs in the $m$-th subnetwork is denoted as $\mathcal{C}_m$.
It is clear that
\begin{align}
\label{non_overlapping_subnetwork_definition}
\qquad\mathcal{C}_m \cap \mathcal{C}_{m'} = \emptyset, \qquad {\forall}m' \neq m,\;m=1,2,\cdots,M,
\end{align}
and
\begin{align}
\label{union_subnetworks}
\bigcup_{m=1}^M\mathcal{C}_m=\mathcal{U} \cup \mathcal{B}.
\end{align}
The corresponding network decomposition $\mathcal{C}$ is denoted as
\begin{align}
\label{network_topology_c_definition}
\mathcal{C}=\{\mathcal{C}_1, \mathcal{C}_2, \cdots, \mathcal{C}_M\}.
\end{align}
As the network is modeled as a weighted undirected bipartite graph, each subnetwork indicated by $\mathcal{C}_m$ can be represented as a subgraph $\mathcal{G}[\mathcal{C}_m]$. Here $\mathcal{G}[\mathcal{C}_m]$ denotes the subgraph of $\mathcal{G}$ induced by $\mathcal{C}_m$, which takes $\mathcal{C}_m$ as its vertex set and contains all the edges in $\mathcal{G}$ that have both endpoints in $\mathcal{C}_m$.    
The clustered cell-free networking problem is then equivalent to a bipartite graph partitioning problem.

Denote the set of UEs and the set of BSs in the $m$-th subnetwork as $\mathcal{U}_m=\mathcal{C}_m\cap\mathcal{U}$ and $\mathcal{B}_m=\mathcal{C}_m\cap\mathcal{B}$, respectively, with $\mathcal{U}_m \cup \mathcal{B}_m=\mathcal{C}_m$.
For the uplink transmission, the received signal at the BSs in the $m$-th subnetwork can be expressed as
\begin{align}
\textbf{y}_m {=} \underbrace{\sum\limits_{u_k{\in}\mathcal{U}_m}\hspace{-2mm}\sqrt{P}\textbf{h}_{\mathcal{B}_m,k}{\cdot}s_k}_{\text{desired signal}} {+} \underbrace{\sum\limits_{u_{k'}{\in}\mathcal{U}\setminus\mathcal{U}_m}\hspace{-3mm}\sqrt{P}\textbf{h}_{\mathcal{B}_m,k'}{\cdot}s_{k'}}_{\text{inter-subnetwork interference}} {+} \textbf{n}_m,
\end{align}
where $\textbf{h}_{\mathcal{B}_m,k}=[h_{lk}]^T_{b_l \in \mathcal{B}_m}$ denotes the $\vert \mathcal{B}_m \vert$ $\times 1$ channel gain vector from UE $u_{k}$ to the BSs in the $m$-th subnetwork.
$h_{lk}$ is the channel gain from UE $u_k$ to BS $b_l$, which is modeled as
\begin{align}
h_{lk}=q_{lk}\cdot g_{lk},
\end{align}
where 
\begin{align}
q_{lk}=d_{lk}^{-\alpha/2}
\end{align}
is the path-loss. $d_{lk}$ denotes the Euclidean distance between BS $b_l$ and UE $u_k$, and $\alpha$ is the path-loss exponent.
$g_{lk} \sim \mathcal{CN}(0,1)$ denotes the corresponding small-scale fading coefficient.
$s_k$ is the data signal transmitted by UE $u_k$ with $\mathbb{E}\{{\vert}s_k{\vert}^2\}=1$.
$P$ is the transmit power of each UE.
$\textbf{n}_m \sim \mathcal{CN}(\textbf{0},N_0\textbf{I}_{\vert \mathcal{B}_m \vert})$ is the additive white Gaussian noise (AWGN) vector.

By normalizing the total system bandwidth into unity, the ergodic capacity of the $m$-th subnetwork is given by (9) \cite{tse_fundamentals_wireless_comm,c2_what_should_future_network_be}, which is shown at the bottom of this page, where $\textbf{H}_m=[\textbf{h}_{\mathcal{B}_m,k}]_{u_k \in \mathcal{U}_m}$ is the $\vert \mathcal{B}_m \vert$ $\times$ $\vert \mathcal{U}_m \vert$ channel gain matrix of the $m$-th subnetwork.
\begin{figure*}[!b]
\hrule
\setcounter{equation}{8} 
\begin{align}
\label{original_cluster_capacity_expr}
C_{sub}(\mathcal{C}_m)=\mathbb{E}\left\{ \log_{2}\det\left(\textbf{I}_{\vert \mathcal{B}_m \vert}+P\left(N_0\textbf{I}_{\vert \mathcal{B}_m \vert}+P\sum\limits_{u_{k'}{\in}\mathcal{U}\setminus\mathcal{U}_m}\textbf{h}_{\mathcal{B}_m,k'}\textbf{h}_{\mathcal{B}_m,k'}^H\right)^{-1}\hspace{-2mm}\textbf{H}_m\textbf{H}_m^H\right)\right\}
\end{align}
\end{figure*}

The sum ergodic capacity of the whole network is then
\begin{align}
C_{sum}=\sum\limits_{m=1}^MC_{sub}(\mathcal{C}_m).
\end{align}
Note that when the number of UEs $K$ and the number of BSs $L$ are large, the equivalent channel gain matrix $\left(N_0\textbf{I}_{\vert \mathcal{B}_m \vert}+P\sum_{u_{k'}{\in}\mathcal{U}\setminus\mathcal{U}_m}\textbf{h}_{\mathcal{B}_m,k'}\textbf{h}_{\mathcal{B}_m,k'}^H\right)^{-1}\textbf{H}_m\textbf{H}_m^H$ asymptotically reduces to a diagonal matrix and the ergodic capacity $C_{sub}(\mathcal{C}_m)$ given in (\ref{original_cluster_capacity_expr}) can be approximated as \cite{c2_what_should_future_network_be}
\begin{align}
\label{approximated_cluster_capacity_expr}
C_{sub}(\mathcal{C}_m)\approx\sum_{b_l{\in}\mathcal{B}_m}\log_2(1+P\lambda_{l}),
\end{align}
where
\begin{align}
\label{lambda_l}
\lambda_{l}=\frac{\sum\limits_{u_k{\in}\mathcal{U}_m}q_{lk}^2}{N_0+P\sum\limits_{u_{k'}{\in}\mathcal{U}{\setminus}\mathcal{U}_m}q_{lk'}^2}.
\end{align}
The main notations used in this paper are listed in Table I for ease of reading.
\begin{table}[!t]
\begin{center}
\caption{Main notations}    
\begin{tabular}{|c || c|} 
\hline
Symbols & Definitions \\
\hline
$K$ & Number of UEs \\
\hline
$L$ & Number of BSs \\
\hline
$\mathcal{U}=\{u_1,u_2,\cdots,u_K\}$ & Set of UEs \\
\hline
$\mathcal{B}=\{b_1,b_2,\cdots,b_L\}$ & Set of BSs \\
\hline
$\mathcal{G}=(\mathcal{V}, \mathcal{E})$ & A graph with vertex set $\mathcal{V}$ and edge set $\mathcal{E}$ \\
\hline
$\textbf{A} \in \mathbb{R}^{(K+L)\times(K+L)}$ & Weighted adjacency matrix of graph $\mathcal{G}$ \\
\hline
$\textbf{W} \in \mathbb{R}^{K \times L}$ & Edge weight matrix of graph $\mathcal{G}$ \\
\hline
$\textbf{D} \in \mathbb{R}^{(K+L)\times(K+L)}$ & Degree matrix of graph $\mathcal{G}$ \\
\hline
$\textbf{L} \in \mathbb{R}^{(K+L)\times(K+L)}$ & Laplacian matrix of graph $\mathcal{G}$ \\
\hline
$M$ & Number of decomposed subnetworks \\
\hline
$\mathcal{C}_m$ & Set of BSs and UEs in $m$-th subnetwork \\
\hline
$\mathcal{C}=\{\mathcal{C}_1, \mathcal{C}_2, \cdots, \mathcal{C}_M\}$ & A network decomposition with $M$ subnetworks \\
\hline
$\mathcal{G}[\mathcal{C}_m]$ & Subgraph of $\mathcal{G}$ induced by subnetwork $\mathcal{C}_m$ \\
\hline
$\mathcal{U}_m=\mathcal{C}_m\cap\mathcal{U}$ & Set of UEs in $\mathcal{C}_m$ \\
\hline
$\mathcal{B}_m=\mathcal{C}_m\cap\mathcal{B}$ & Set of BSs in $\mathcal{C}_m$ \\
\hline
$\textbf{y}_m$ & Received signal at the BSs in $\mathcal{C}_m$ \\
\hline
$\textbf{h}_{\mathcal{B}_m,k}=[h_{lk}]^T_{b_l \in \mathcal{B}_m}$ & \thead{Channel gain vector \\ from UE $u_{k}$ to the BSs in $\mathcal{C}_m$}\\
\hline
$h_{lk}=q_{lk}\cdot g_{lk}$ & Channel gain from UE $u_{k}$ to BS $b_{l}$ \\
\hline
$q_{lk}=d_{lk}^{-\alpha/2}$ & Path-loss \\
\hline
$d_{lk}$ & Euclidean distance between BS $b_l$ and UE $u_k$ \\
\hline
$\alpha$ & Path-loss exponent \\
\hline
$g_{lk} \sim \mathcal{CN}(0,1)$ & Small-scale fading coefficient \\
\hline
$s_k$ & Data signal transmitted by UE $u_k$ \\
\hline
$P$ & Transmit power of each UE \\
\hline
$\textbf{n}_m$ & Additive white Gaussian noise (AWGN) vector \\
\hline
$C_{sub}(\mathcal{C}_m)$ & Ergodic capacity of $\mathcal{C}_m$ \\
\hline
$C_{sum}$ & Sum ergodic capacity of the network \\
\hline
$K_{max}$ & \thead{Maximum allowable number of UEs \\ in each subnetwork} \\
\hline
\end{tabular}
\end{center}
\end{table}

\subsection{Problem Formulation}
\label{Problem_Formulation}
In this paper, we aim to optimize the clustered cell-free networking scheme by jointly optimizing the number of subnetworks $M$ and the corresponding network decomposition $\mathcal{C}=\{\mathcal{C}_1, \mathcal{C}_2, \cdots, \mathcal{C}_M\}$ with the objective of maximizing the sum ergodic capacity of the network\footnote{In this paper, we aim to maximize the sum ergodic capacity instead of the instantaneous sum capacity within a given coherence interval because maximizing the instantaneous sum capacity would lead to the clustered cell-free networking result, i.e., BS-UE association, varying with small-scale fading.
As the small-scale fading is fast-varying, frequent handover along with huge signaling overhead and high computational costs are required, which are unaffordable in practice.
This motivates us to maximize the sum ergodic capacity that averages out the effect of small-scale fading.} under the joint processing constraint.
Since the joint processing complexity in a subnetwork increases as the number of UEs in it increases and the number of UEs that could be served by a BS is usually limited in practice \cite{static_clustering_cran, dl_ra_cf_mmimo_uc_clustering,greedy_mu_scheduling_clustered_cf_mmimo}, we propose to constrain the joint processing complexity by setting a maximum allowable number of UEs $K_{max}$ for each single subnetwork.\footnote{Here, the joint processing complexity in a subnetwork refers to the computational costs and signaling overhead incurred by the joint processing among BSs, including but not limited to the computational costs for channel measurement and jointly detecting UEs' data, and the signaling overhead for channel measurement and dynamic BS-UE association, etc.
For the uplink transmission, the computational complexity required for jointly detecting UEs' data increases with the number of UEs.
Moreover, the processing capability of a BS is usually limited and the number of UEs a BS can serve is also limited due to the limited radio frequency resources.
Therefore, in this paper, we constrain that the number of UEs in each subnetwork is bounded by a preset limit $K_{max}$ for the sake of keeping the joint processing complexity at an acceptable level.
Note that restricting the number of BSs is also of practical importance, which could further be included in our future work.}
The clustered cell-free networking problem can be mathematically formulated as
\begin{subequations}
\label{ul_capacity_maximization_problem_origin}
\begin{align}
\label{ul_capacity_objctive_origin}
\;\;\mathcal{P}1:\;\max\limits_{\mathcal{C}=\{\mathcal{C}_1, \mathcal{C}_2, \cdots, \mathcal{C}_M\},M}\;C_{sum}=\sum\limits_{m=1}^MC_{sub}(\mathcal{C}_m)&&
\end{align}
\vspace{-4mm}
\begin{alignat}{2}
\label{nonoverlapping_constraint}
{\qquad}{\quad}\text{s.t.}{\quad}&\mathcal{C}_m \cap \mathcal{C}_{m'} = \emptyset, &\forall m' \neq m,\;\forall m,\\
\label{all_node_clustered_constraint}
&\bigcup\limits_{m=1}^M\mathcal{C}_m=\mathcal{U} \cup \mathcal{B}, \\
\label{ue_per_cluster_constraint}
&\vert\mathcal{C}_m\cap\mathcal{U}\vert\leq K_{max}, &\forall m, \\
\label{bs_per_cluster_constraint}
&\mathcal{C}_m\cap\mathcal{B}\neq\emptyset, &\forall m,
\end{alignat}
\end{subequations}
where (\ref{nonoverlapping_constraint}) and (\ref{all_node_clustered_constraint}) follow the definition of clustered cell-free networking given in (\ref{non_overlapping_subnetwork_definition}) and (\ref{union_subnetworks}) that the network is decomposed into $M$ non-overlapping subnetworks. 
(\ref{ue_per_cluster_constraint}) is the joint processing constraint that the number of UEs in each subnetwork is no larger than $K_{max}$.
(\ref{bs_per_cluster_constraint}) ensures that each subnetwork contains at least one BS to serve the UEs in it for the sake of service outage avoidance.

\section{Clustered Cell-Free Networking for Sum Ergodic Capacity Maximization}
\label{Proposed_Solution}
To maximize the sum ergodic capacity of the network, it can be seen from  problem $\mathcal{P}1$ that the number of subnetworks $M$ and the corresponding network decomposition $\mathcal{C}=\{\mathcal{C}_1, \mathcal{C}_2, \cdots, \mathcal{C}_M\}$ need to be jointly optimized. 
However, the network decomposition $\mathcal{C}$ and the number of subnetworks $M$ are coupled with each other, i.e. the number of subnetworks $M$ affects the corresponding network decomposition $\mathcal{C}$, making the formulated combinatorial clustered cell-free networking problem difficult to solve.
In fact, Theorem 1 proves that problem $\mathcal{P}1$ is NP-hard.
\begin{thm}
\label{np_hardness_throrem}
Problem $\mathcal{P}1$ in (\ref{ul_capacity_maximization_problem_origin}) is NP-hard.
\end{thm}
\begin{proof}
See Appendix A.
\end{proof}

Due to the NP-hardness of the joint optimization problem $\mathcal{P}1$, in the following, we will first determine the optimal number of subnetworks $M^*$ and then reformulate problem $\mathcal{P}1$ into a solvable form.

\subsection{Optimal Number of Subnetworks $M^*$}
\label{Preliminaries}
A straightforward way to find out the optimal number of subnetworks $M^*$ of problem $\mathcal{P}1$ has two steps: 1) figure out the optimal network decomposition that maximizes the sum ergodic capacity of the network for each possible number of subnetworks $M$, and 2) exhaustively search over all the possible values of $M$.
However, the computational complexity of the above brute-force search method is high when the total number of UEs $K$ is large, as the number of subnetworks $M$ could vary from the smallest integer that satisfies constraint (\ref{ue_per_cluster_constraint}) to the total number of UEs $K$.
Theorem \ref{monotone_property_throrem} shows that the maximum ergodic capacity of the whole network decreases monotonically as the number of subnetworks $M$ increases.
\begin{thm}
\label{monotone_property_throrem}
For any number of subnetwork $M$, the inequality $\max\limits_{\{\mathcal{C}_1,\mathcal{C}_2,\cdots,\mathcal{C}_{M+1}\}}\sum\limits_{m=1}^{M+1}\hspace{-2mm}C_{sub}(\mathcal{C}_m) \leq \max\limits_{\{\mathcal{C}_1,\mathcal{C}_2,\cdots,\mathcal{C}_M\}}\sum\limits_{m=1}^M\hspace{-2mm}C_{sub}(\mathcal{C}_m)$ holds.
\end{thm}
\begin{proof}
See Appendix B.
\end{proof}
According to Theorem \ref{monotone_property_throrem}, in order to maximize the sum ergodic capacity of the network, the number of subnetworks $M$ needs to be as small as possible.
Constrained by (\ref{ue_per_cluster_constraint}), the optimal number of subnetworks $M^*$ should be the smallest integer satisfying that the number of UEs in each subnetwork is no larger than the preset limit $K_{max}$.
The optimal
number of subnetworks $M^*$ can be then obtained in a closed-form expression as
\begin{align}
\label{M_definition}
M^*=
\left\lceil
\frac{K}{K_{max}} 
\right\rceil,
\end{align}
where $\lceil\cdot\rceil$ is the ceiling operator.

Given the optimal number of subnetworks $M^*$ in (\ref{M_definition}), the clustered cell-free networking problem $\mathcal{P}1$ in (\ref{ul_capacity_maximization_problem_origin}) reduces to
\begin{subequations}
\label{ul_capacity_maximization_problem_with_M*}
\begin{align}
\label{ul_capacity_objctive_with_M*}
\mathcal{P}2:\hspace{-2mm}&\max\limits_{\{\mathcal{C}_{1|M^*}, \mathcal{C}_{2|M^*}, \cdots, \mathcal{C}_{M^*|M^*}\}}\hspace{-2mm}C_{sum}{=}\hspace{-1mm}\sum\limits_{m=1}^{M^*}C_{sub}(\mathcal{C}_{m|M^*})&&
\end{align}
\begin{alignat}{2}
\label{nonoverlapping_constraint_with_M*}
\text{s.t.}{\quad}&\mathcal{C}_{m|M^*} \cap \mathcal{C}_{m'|M^*} = \emptyset, &\;\forall m' \neq m;\;\forall m, \\
\label{all_node_clustered_constraint_with_M*}
&\bigcup\limits_{m=1}^{M^*}\mathcal{C}_{m|M^*}=\mathcal{U} \cup \mathcal{B}, \\
\label{ue_per_cluster_constraint_with_M*}
&\vert\mathcal{C}_{m|M^*}\cap\mathcal{U}\vert\leq K_{max}, &\forall m, \\
\label{bs_per_cluster_constraint_with_M*}
&\mathcal{C}_{m|M^*}\cap\mathcal{B}\neq\emptyset, &\forall m.
\end{alignat}
\end{subequations}
However, the objective function in (\ref{ul_capacity_objctive_with_M*}) and the contraints in (\ref{nonoverlapping_constraint_with_M*})-(\ref{bs_per_cluster_constraint_with_M*}) are still functions of sets, none of which is in an algebraic form that can be handled easily.
In the next subsection, we will transform the optimization problem $\mathcal{P}2$ into a convex optimization problem.

\subsection{Problem Reformulation}
\label{Problem_Reformulation}
By substituting (\ref{approximated_cluster_capacity_expr}) and (\ref{lambda_l}) into (\ref{ul_capacity_objctive_with_M*}), the sum ergodic capacity $C_{sum}$ can be obtained as
\begin{align}
\label{transformed_objective}
C_{sum}{=}\hspace{-1.5mm}\sum\limits_{b_l{\in}\mathcal{B}}\log_2a_l{-}\hspace{-1.5mm}\sum\limits_{m=1}^{M^*}\sum\limits_{b_l{\in}\mathcal{B}_{m\vert M^*}}\hspace{-4mm}\log_2\hspace{-1.5mm}\left(\hspace{-1.5mm}N_0{+}P\hspace{-5mm}\sum\limits_{u_{k'}{\in}\mathcal{U}\setminus\mathcal{U}_{m\vert M^*}}\hspace{-5mm}q_{lk'}^2\hspace{-1.5mm}\right),
\end{align}
where 
\begin{align}
\label{a_l}
a_l=N_0+P\sum\limits_{u_k{\in}\mathcal{U}}q_{lk}^2 
\end{align}
is a constant for a given BS $b_l$.
It can be observed from (\ref{transformed_objective}) that the second term is a double summation of a set of logarithmic functions with another summation inside each logarithmic function, which is difficult to handle. To obtain a tractable form of the ergodic capacity, let us apply Jensen's Inequality on the second term of (\ref{transformed_objective}). A lower-bound of the ergodic capacity, $C_{sum\_lb}$, can be then obtained as
\begin{align}
\label{lower_bound_objective}
&C_{sum}{\geq}C_{sum\_lb} \nonumber \\
&{=}\hspace{-2mm}\sum\limits_{b_l{\in}\mathcal{B}}\hspace{-1mm}\log_2a_l{-}L\hspace{-0.5mm}\log_2\hspace{-1.5mm}\left(\hspace{-1.5mm}N_0{+}\frac{P}{L}\hspace{-1.5mm}\sum\limits_{m=1}^{M^*}\sum\limits_{b_l{\in}\mathcal{B}_{m\vert M^*}}\sum\limits_{u_{k'}{\in}\mathcal{U}\setminus\mathcal{U}_{m\vert M^*}}\hspace{-6.5mm}q_{lk'}^2\hspace{-1mm}\right)\hspace{-1mm}.
\end{align}

By modeling the edge weight between UE $u_k$ and BS $b_l$, $w_{kl}$, as the corresponding path-loss, that is,
\begin{align}
\label{edge_weight_equals_ls_fading}
w_{kl}=q_{lk}^2,
\end{align}
the term $\sum_{m=1}^{M^*}\sum_{b_l{\in}\mathcal{B}_{m\vert M^*}}\sum_{u_{k'}{\in}\mathcal{U}\setminus\mathcal{U}_{m\vert M^*}}q_{lk'}^2$ in (\ref{lower_bound_objective}) can be rewritten as
\begin{align}
\label{sum_form_equals_sumcut}
\sum\limits_{m=1}^{M^*}\sum\limits_{b_l{\in}\mathcal{B}_{m\vert M^*}}\sum\limits_{u_{k'}{\in}\mathcal{U}\setminus\mathcal{U}_{m\vert M^*}}q_{lk'}^2=\frac{1}{2}\sum\limits_{m=1}^{M^*}\text{cut}(\mathcal{C}_{m|M^*}),
\end{align}
where
\begin{align}
\label{cut_function_definition}
\text{cut}(\mathcal{C}_{m|M^*}){=}\hspace{-5mm}\sum\limits_{b_l{\in}\mathcal{B}_{m\vert M^*}}\sum\limits_{u_{k}{\in}\mathcal{U}\setminus\mathcal{U}_{m\vert M^*}}\hspace{-6mm}w_{kl}+\hspace{-3mm}\sum\limits_{u_k{\in}\mathcal{U}_{m|M^*}}\sum\limits_{b_{l}{\in}\mathcal{B}\setminus\mathcal{B}_{m|M^*}}\hspace{-6mm}w_{kl}
\end{align}
is the cut function of subgraph $\mathcal{G}[\mathcal{C}_{m|M^*}]$.

By replacing the objective function of problem $\mathcal{P}2$ with its lower-bound obtained in (\ref{lower_bound_objective}),\footnote{The effectiveness of replacing (\ref{ul_capacity_objctive_with_M*}) with its lower-bound in (\ref{lower_bound_objective}) will be demonstrated in detail in Section \ref{Effectiveness_of_Lower_Bound_Approximation}.} and combining with (\ref{edge_weight_equals_ls_fading})-(\ref{sum_form_equals_sumcut}), the optimization problem $\mathcal{P}2$ in (\ref{ul_capacity_maximization_problem_with_M*}) can be transformed into the following form
\begin{subequations}
\label{ul_capacity_maximization_problem_lb}
\begin{align}
\label{ul_capacity_objective_lb}
\mathcal{P}3:\;&\min\limits_{\{\mathcal{C}_{1|M^*}, \mathcal{C}_{2|M^*}, \cdots, \mathcal{C}_{M^*|M^*}\}}\;\sum\limits_{m=1}^{M^*}\text{cut}(\mathcal{C}_{m|M^*}) \\
\label{ul_capacity_constraints_lb}
&\text{s.t.}\quad\text{(\ref{nonoverlapping_constraint_with_M*}), (\ref{all_node_clustered_constraint_with_M*}), (\ref{ue_per_cluster_constraint_with_M*}), (\ref{bs_per_cluster_constraint_with_M*})}.
\end{align}
\end{subequations}
It can be seen from (\ref{ul_capacity_maximization_problem_lb}) that problem $\mathcal{P}3$ is a $M^*$-way graph partitioning problem with the objective of minimizing the sumcut.
However, compared to the canonical min $k$-cut problem, there are two extra constraints, i.e.,  (\ref{ue_per_cluster_constraint_with_M*}) and (\ref{bs_per_cluster_constraint_with_M*}), in problem $\mathcal{P}3$, making it difficult to apply existing min $k$-cut algorithms to solve it.
In the following, we will transform problem $\mathcal{P}3$ into a convex form.

To indicate which subnetwork vertex $v_i$ belongs to, let us introduce a binary variable $x_{i,m}$.
Let $x_{i,m}=1$ if vertex $v_i$ is in the $m$-th subnetwork; otherwise $x_{i,m}=0$.
That is,
\begin{align}
\label{binary_variable_definition}
x_{i,m}=
\left\lbrace
\begin{array}{cl}
1, &\text{if }v_i \in \mathcal{C}_m,\\
0, &\text{if }v_i \notin \mathcal{C}_m.
\end{array}
\right.
\end{align}
The network decomposition $\mathcal{C}{=}\{\mathcal{C}_{1|M^*}, \mathcal{C}_{2|M^*},{\cdots},\mathcal{C}_{M^*|M^*}\}$ can be then determined by an $(K+L) \times M^*$ decision matrix $\textbf{X}=[\textbf{x}_{1|M^*},\textbf{x}_{2|M^*},{\cdots},\textbf{x}_{M^*|M^*}]$, and the objective function in (\ref{ul_capacity_objective_lb}) can be rewritten as
\begin{align}
\label{cut_to_quad_form}
\sum\limits_{m=1}^{M^*}\text{cut}(\mathcal{C}_{m|M^*})=\sum\limits_{m=1}^{M^*}(\textbf{x}_{m|M^*})^T\textbf{L}\textbf{x}_{m|M^*},
\end{align}
where
\begin{align}
\label{L_definition}
\textbf{L}=\textbf{D}-\textbf{A}
\end{align}
is the Laplacian matrix and $\textbf{D}=diag\{d_1,d_2,{\cdots},d_{K+L}\}$ denotes the degree matrix of graph $\mathcal{G}$ with $d_i=\sum\limits_{v_j \in \mathcal{V}}a_{ij}$. $a_{ij}$ is the $i$-th row and $j$-th column entry of adjacency matrix $\textbf{A}$ given in (\ref{adjacency_matrix_definition}).

Similarly, constraints in (\ref{nonoverlapping_constraint_with_M*}) and (\ref{all_node_clustered_constraint_with_M*}) can be rewritten as 
\begin{align}
\label{nonoverlapping_constraint_vec}
&\sum\limits_{m=1}^{M^*}\textbf{e}_i^T\textbf{x}_{m|M^*}=1, \qquad i=1,2,{\cdots},K+L,
\end{align}
and
\begin{align}
\label{all_node_clustered_constraint_vec}
&\sum\limits_{m=1}^{M^*}\textbf{1}^T\textbf{x}_{m|M^*}=K+L,
\end{align}
respectively, where $\textbf{e}_i$ is the $(K+L)\times1$ unit vector with the $i$-th entry being one and the other entries being zeros. $\textbf{1}$ is the $(K+L)\times1$ vector of ones.
By defining an auxiliary vector
\begin{align}
\textbf{t}=[\underbrace{1,{\cdots},1}_{\text{$K$ 1s}},\underbrace{0,{\cdots},0}_{\text{$L$ 0s}}]^T,
\end{align}
the constraints in (\ref{ue_per_cluster_constraint_with_M*}) and (\ref{bs_per_cluster_constraint_with_M*}) can be rewritten as
\begin{align}
\label{ue_per_cluster_constraint_vec}
&\textbf{t}^T\textbf{x}_{m|M^*} \leq K_{max}, \qquad m=1,2,{\cdots},M^*,
\end{align}
and
\begin{align}
\label{bs_per_cluster_constraint_vec}
&(\textbf{1}-\textbf{t})^T\textbf{x}_{m|M^*} \geq 1, \qquad m=1,2,{\cdots},M^*,
\end{align}
respectively.

By combining (\ref{cut_to_quad_form})-(\ref{bs_per_cluster_constraint_vec}), problem $\mathcal{P}3$ in (\ref{ul_capacity_maximization_problem_lb}) is  equivalently transformed to
\begin{subequations}
\label{ul_capacity_maximization_problem_vec}
\begin{align}
\label{ul_capacity_objective_vec}
\mathcal{P}4:\;&\min\limits_{\textbf{X}=[\textbf{x}_{1|M^*},\textbf{x}_{2|M^*},{\cdots},\textbf{x}_{M^*|M^*}]}\;\sum\limits_{m=1}^{M^*}(\textbf{x}_{m|M^*})^T\textbf{L}\textbf{x}_{m|M^*} \\
&\text{s.t.}\quad\text{(\ref{nonoverlapping_constraint_vec}), (\ref{all_node_clustered_constraint_vec}), (\ref{ue_per_cluster_constraint_vec}), (\ref{bs_per_cluster_constraint_vec})}, \nonumber \\
\label{binary_variable_constraint_vec}
&\qquad\; x_{i,m}\in\{0,1\},\quad \forall i,\;\forall m.
\end{align}
\end{subequations}

Since Laplacian matrix $\textbf{L}$ is semi-definite positive \cite{algebraic_graph_theory_book}, the quadratic form $(\textbf{x}_{m|M^*})^T\textbf{L}\textbf{x}_{m|M^*}$ is convex.
As a result, the objective function in (\ref{ul_capacity_objective_vec}) is convex as the summation of a set of convex functions is still convex.
Moreover, constraints (\ref{nonoverlapping_constraint_vec})-(\ref{bs_per_cluster_constraint_vec}) are all affine and thus convex.
Therefore, problem $\mathcal{P}4$ in (\ref{ul_capacity_maximization_problem_vec}) is a convex integer programming problem, whose global optimum can be obtained via the branch-and-bound method \cite{branch_and_bound} by utilizing existing convex solvers such as MOSEK toolbox \cite{mosek}.
The detailed description is presented in Algorithm 1.

\begin{algorithm}[!t]
\caption{Branch-and-bound based clustered cell-free networking scheme}
\begin{algorithmic}[1]
\REQUIRE Bipartite graph $\mathcal{G}$, number of UEs $K$, number of BSs $L$, weight matrix $\textbf{W}$, maximum allowable per-subnetwork UE number $K_{max}$
\PRINT $M^*=\left\lceil\frac{K}{K_{max}}\right\rceil$;
\STATE Compute $\textbf{A}$ and $\textbf{L}$ according to (2) and (25); 
\STATE Invoke MOSEK toolbox to solve problem $\mathcal{P}4$ in (31) to obtain $\mathcal{C}_{1}, \mathcal{C}_{2}, \cdots, \mathcal{C}_{M^*}$;
\ENSURE $\mathcal{C}=\{\mathcal{C}_{1}, \mathcal{C}_{2}, \cdots, \mathcal{C}_{M^*}\}$
\end{algorithmic}
\label{branch_and_bound_based_scheme}
\end{algorithm}

\subsection{Complexity Analysis}
\label{Complexity_Analysis}
The general idea of the branch-and-bound approach is to explore all the binary optimization variables iteratively.
Specifically, in each iteration, the solution space is divided into a set of small subspaces, and a relaxed lower-bound and an achievable upper-bound of the optimal value are computed in each subspace.
The computational complexity in the $n$-th iteration of the branch-and-bound approach is in the order of
\begin{align}
\label{complexity_integer_programming}
\sum_{j{\in}\{0,1\}}\mathcal{O}(R_j^{(n),lb})+\mathcal{O}(R_j^{(n),ub}),
\end{align}
where $R_j^{(n),lb}$ and $R_j^{(n),ub}$ denote the complexity of computing the relaxed lower-bound and the achievable upper-bound in the $n$-th iteration, respectively.

Assume that the solution space is divided into two child subspaces with the set of unfixed binary variables $\mathcal{Q}_{ju}^{(n)}$, $j\in\{0,1\}$ in the $n$-th iteration.
By relaxing all the unfixed binary variables in $\mathcal{Q}_{ju}^{(n)}$ to continuous variables, the relaxed lower-bound can be obtained by solving a convex quadratic programming problem with computational complexity \cite{lectures_modern_convex_opt}
\begin{align}
\label{o_c_i_lb}
R_j^{(n),lb}=&\sum_{j{\in}\{0,1\}}\hspace{-2mm}\ln\hspace{-1mm}\left(\hspace{-1mm}\frac{1}{\epsilon}\hspace{-1mm}\right)\hspace{-1mm}\sqrt{4(K+L)M^*{+}1}\hspace{-0.5mm} \nonumber \\
&\left(\hspace{-0.5mm}\vert\mathcal{Q}_{ju}^{(n)}\vert^3{+}(4M^*{+}1)\hspace{-0.5mm}\left(\hspace{-0.5mm}\vert\mathcal{Q}_{ju}^{(n)}\vert^2{+}\vert\mathcal{Q}_{ju}^{(n)}\vert\hspace{-0.5mm}\right)\hspace{-0.5mm}\right).
\end{align}
By rounding all the entries in the solution to the above relaxed convex quadratic programming problem to binary values, a feasible upper-bound of the global optimum of problem $\mathcal{P}4$ can be obtained by substituting the binary variables into (\ref{ul_capacity_objective_vec}), which requires $(K+L)^2-1$ additions and $(K+L)(K+L+1)$ multiplications.
The computational complexity of obtaining the feasible upper-bound, $R_j^{(n),ub}$, is thus
\begin{align}
\label{o_c_i_ub}
R_j^{(n),ub}=2(K+L)^2+K+L-1.
\end{align}

By combining (\ref{complexity_integer_programming})--(\ref{o_c_i_ub}), the computational complexity in the $n$-th iteration of the branch-and-bound approach for solving problem $\mathcal{P}4$ to obtain the network decomposition result can be obtained as
\begin{align}
\label{complexity_nth_iteration_complete}
&\sum_{j{\in}\{0,1\}}\ln\left(\frac{1}{\epsilon}\right)\sqrt{4(K+L)M^*{+}1}\left(\vert\mathcal{Q}_{ju}^{(n)}\vert^3{+}(4M^*{+}1)\cdot\right. \nonumber \\
&\left.\left(\vert\mathcal{Q}_{ju}^{(n)}\vert^2{+}\vert\mathcal{Q}_{ju}^{(n)}\vert\right)\right)+2(K+L)^2+K+L-1.
\end{align}
It can be seen from (\ref{complexity_nth_iteration_complete}) that the computational complexity closely depends on the cardinality $\vert\mathcal{Q}_{ju}^{(n)}\vert$, i.e., the number of unfixed binary variables.
As the binary variables are determined gradually as the algorithm iterates, the first iteration has the maximum $\vert\mathcal{Q}_{ju}^{(n)}\vert$, i.e.,
\begin{align}
\label{vert_q_j_u_1_vert}
\vert\mathcal{Q}_{ju}^{(n)}\vert\leq\vert\mathcal{Q}_{ju}^{(1)}\vert=(K+L) \cdot M^*-1.
\end{align}
By combining (\ref{M_definition}), (\ref{complexity_nth_iteration_complete}) and (\ref{vert_q_j_u_1_vert}), an upper-bound of the per-iteration computational complexity of the branch-and-bound approach can be obtained, which is in the order of
\begin{align}
\label{first_iteration_complexity}
\mathcal{O}\left(\ln\left(\frac{1}{\epsilon}\right)\left\lceil\frac{K}{K_{max}}\right\rceil^{3.5}\left(K+L\right)^{3.5}\right).
\end{align}
Therefore, the total computational complexity of the branch-and-bound approach is 
\begin{align}
\label{total_complexity_ub}
\mathcal{O}\left(N_{itr}\ln\left(\frac{1}{\epsilon}\right)\left\lceil\frac{K}{K_{max}}\right\rceil^{3.5}\left(K+L\right)^{3.5}\right),
\end{align}
where $N_{itr}$ is the number of iterations, or in other words, the number of visited nodes in the binary tree.

In the cases where the maximum number of UEs that can be jointly served in each subnetwork, $K_{max}$, is small, the computational complexity scales with the number of UEs $K$ and the number of BSs $L$ in the order of $\mathcal{O}\left(\ln\left(\frac{1}{\epsilon}\right)K^{3.5}\left(K+L\right)^{3.5}\right)$, which increases dramatically as the number of UEs $K$ increases, indicating that the branch-and-bound based solution might not be suitable for dense scenarios with a massive number of UEs. 
Moreover, it should be noted that since the branch-and-bound approach randomly chooses an unfixed variable in each iteration, in the worst case, the branch-and-bound method has the same complexity as the exhaustive search that traverses all the binary variables in the solution space, i.e., $N_{itr}=2^{(K+L)\left\lceil{K}/{K_{max}}\right\rceil}$.
In other words, the branch-and-bound approach cannot be guaranteed to terminate in polynomial time.
Therefore, it is highly desirable to develop a new clustered cell-free networking scheme with lower computational complexity.

\begin{figure}[!t]
\centering
\subfloat[Invalid network decomposition]{
\includegraphics[width=0.24\textwidth]{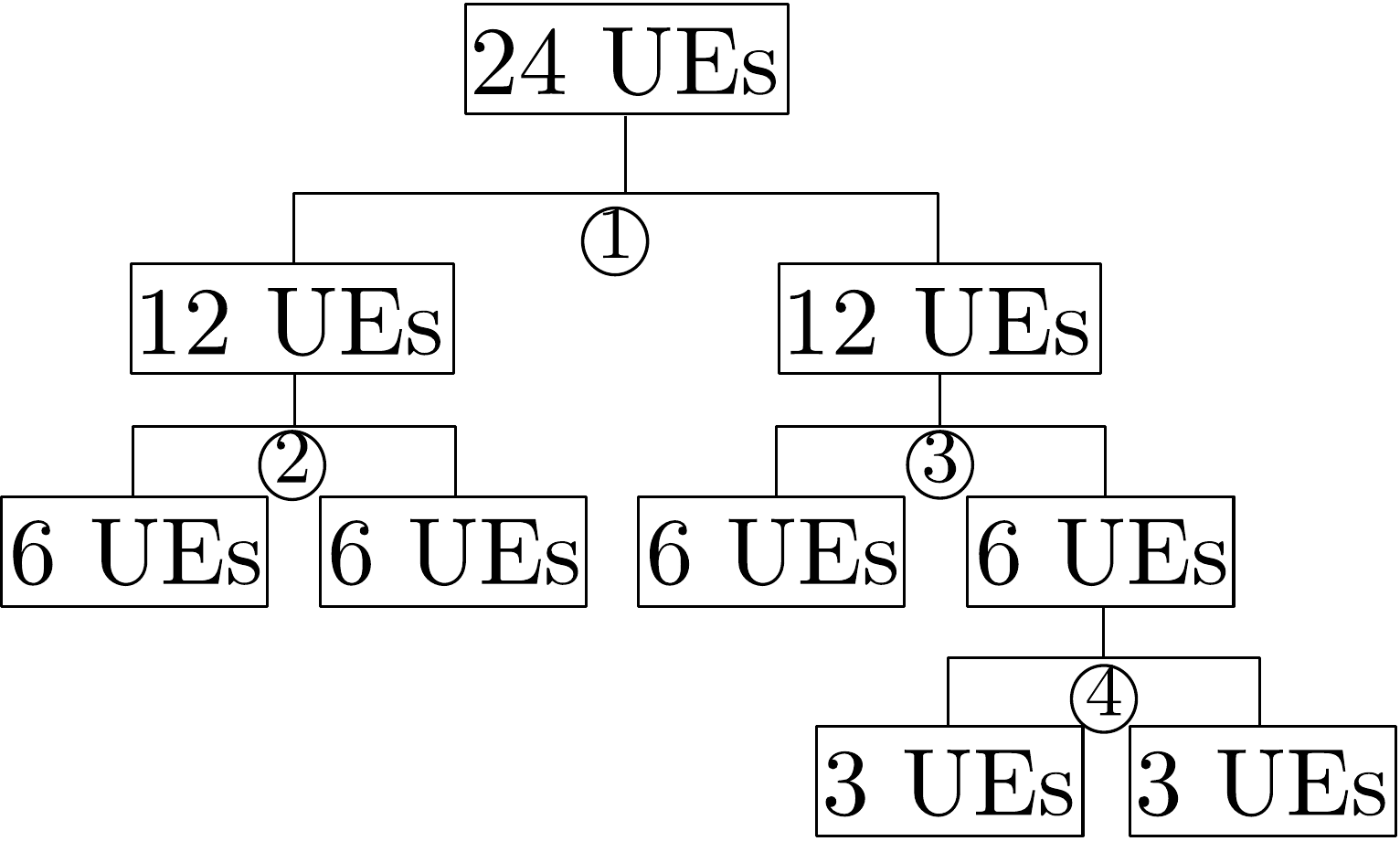}
\label{infeasible_network_decomposition}}
\subfloat[Valid network decomposition]{
\includegraphics[width=0.24\textwidth]{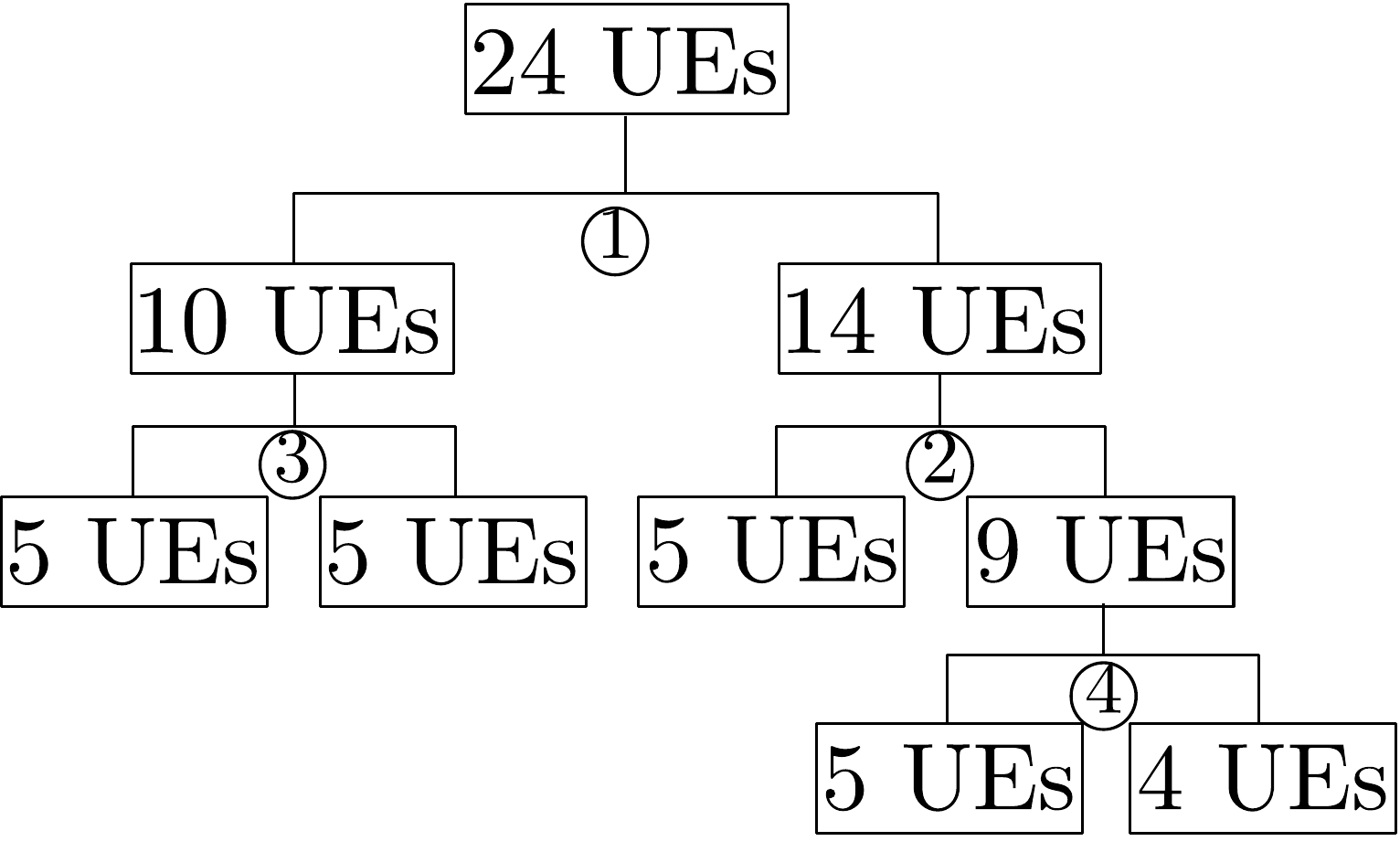}
\label{feasible_network_decomposition}}
\caption{Toy examples of network decomposition with different bisection strategies. Subnetworks are presented by rectangles. Numbers with circles denote iteration indices. $K=24$, $K_{max}=5$, $M^*=5$.}
\label{illustration_different_bisection_strategies}
\end{figure}

\section{Bisection Clustered Cell-Free networking}
\label{Suboptimal_Low_Complexity_Algorithm}
To reduce the computational complexity, a bisection clustered cell-free networking algorithm will be proposed in this section.

\subsection{Bisection Clustered Cell-Free Networking Algorithm}
\label{Hierarchical_Bisection_Clustered_Cell_free_Networking_Algorithm}
As discussed in Section \ref{Complexity_Analysis}, the computational complexity of directly adopting the branch-and-bound approach to solve the clustered cell-free networking problem $\mathcal{P}4$ is upper-bounded by (\ref{total_complexity_ub}), which grows sharply as the number of UEs $K$ or the number of BSs $L$ in the network increases.
An interesting observation from (\ref{total_complexity_ub}) is that when the maximum allowable per-subnetwork UE number $K_{max}$ is large and comparable to the total number of UEs $K$, the computational complexity is significantly reduced, which indicates that decomposing the network into two subnetworks, i.e., $\left\lceil\frac{K}{K_{max}}\right\rceil=2$ has the lowest computational complexity.

Motivated by the above observation, instead of decomposing the network into $M^{*}$ subnetworks directly using the branch-and-bound approach,  we propose to bisect the network hierarchically to reduce the computational complexity.
Specifically, the subnetwork containing the largest number of UEs is decomposed into two subnetworks in each iteration  and the algorithm terminates until $M^*$ subnetworks are obtained.
The key difficulty of bisecting the network into $M^{*}$ subnetworks lies in meeting the following two requirements simultaneously:
\begin{enumerate}
\item Strongly interfered UEs and BSs are in the same subnetwork such that joint processing can be performed inside each subnetwork independently for intra-subnetwork interference mitigation and the inter-subnetwork interference can be treated as noise, i.e., the clustered cell-free networking principle;
\item The number of UEs in each subnetwork is no larger than $K_{max}$, i.e.,  constraint (\ref{ue_per_cluster_constraint_vec}).
\end{enumerate}

As UEs close to each other lead to strong mutual interference, BSs and UEs in the same subnetwork should be geographically close to each other.
Given the optimal number of subnetworks $M^{*}$, to avoid distant UEs and/or BSs being forced into one subnetwork, we propose to bisect the subnetwork into two balanced subnetworks in each iteration, i.e., the bisected subnetworks should contain the same or similar number of UEs.
However, simply bisecting the subnetwork containing the largest number of UEs into two subnetworks with the same number of UEs in each iteration to obtain $M^*$ subnetworks could lead to invalid network decomposition that does not satisfy the joint processing constraint (\ref{ue_per_cluster_constraint_vec}).
A toy example of invalid network decomposition is provided in Fig. \ref{illustration_different_bisection_strategies}\subref{infeasible_network_decomposition}, where the subnetwork containing the largest number of UEs is equally bisected in each iteration.
It can be seen that when $M^{*}$ subnetworks are obtained and the algorithm has terminated, there are still multiple subnetworks with more than $K_{max}$ UEs, i.e. the final network decomposition is invalid.
Therefore, the bisection strategy in each iteration should be carefully designed such that the constraint (\ref{ue_per_cluster_constraint_vec}) is always satisfied and the resulting network decomposition is always feasible.

As constraint (\ref{ue_per_cluster_constraint_vec}) holds when $M^*-1$ out of $M^{*}$ subnetworks contain exactly $K_{max}$ UEs, in each iteration, we propose to partition the subnetwork having the largest number of UEs  into two balanced subnetworks with one of them having an integral multiple of $K_{max}$ UEs.
Denote the largest number of UEs in the subnetworks after the $(n-1)$-th iteration as $K^{(n)}$. The number of UEs $K_{1}^{(n)}$ targeted for one of the bisected subnetworks in the $n$-th iteration can be chosen as
\begin{align}
\label{kth_1_n}
K_{1}^{(n)}=K_{max}\cdot\left\lfloor\frac{1}{2}\cdot\left\lceil\frac{K^{(n)}}{K_{max}}\right\rceil\right\rfloor,
\end{align}
where $\lfloor\cdot\rfloor$ is the floor operator.
The number of UEs $K_{2}^{(n)}$ for the other bisected subnetwork is then given by
\begin{align}
\label{kth_2_n}
K_{2}^{(n)}=K^{(n)}-K_{1}^{(n)}.
\end{align}
A toy example is demonstrated in Fig. \ref{illustration_different_bisection_strategies}\subref{feasible_network_decomposition} to illustrate the proposed bisection strategy.

In addition, to guarantee that each of the finally obtained $M^{*}$ subnetworks contains at least one BS to avoid service outage, i.e., satisfying constraint (\ref{bs_per_cluster_constraint_vec}), the number of BSs in each subnetwork appearing in the whole iterative process should be equal to or larger than the number of subnetworks into which it can further be partitioned.
For a subnetwork containing $K_{m}^{(n)}, m=1, 2$, UEs in the $n$-th iteration, it will be partitioned into $\left\lceil\frac{K_{m}^{(n)}}{K_{max}}\right\rceil$ subnetworks as one subnetwork can have at most $K_{max}$ UEs.

Denote the subgraph representing the subnetwork containing the largest number of UEs in the $n$-th iteration as $\mathcal{G}^{(n)}$, and its adjacency matrix, degree matrix and Laplacian matrix as $\textbf{A}^{(n)}$, $\textbf{D}^{(n)}$ and $\textbf{L}^{(n)}$, respectively.
The bisection of the subnetwork represented by $\mathcal{G}^{(n)}$ can be obtained by solving the following convex integer problem
\begin{subequations}
\label{bisection_problem_vec}
\begin{align}
\label{bisection_objective_vec}
&\mathcal{P}5^{(n)}:\min\limits_{\textbf{X}^{(n)}=[\textbf{x}_{1|2},\textbf{x}_{2|2}]}\;\sum\limits_{m=1}^{2}\textbf{x}_{m|2}^T\textbf{L}^{(n)}\textbf{x}_{m|2} \\
\label{nonoverlapping_constraint_bisection}
&\quad\text{s.t.}\; \sum\limits_{m=1}^{2}(\textbf{e}_i^{(n)})^T\textbf{x}_{m|2}{=}1,~ i{=}1,2,{\cdots},K^{(n)}{+}L^{(n)}, \\
\label{all_node_clustered_constraint_bisection}
&\quad\quad~\sum\limits_{m=1}^{2}(\textbf{1}^{(n)})^T\textbf{x}_{m|2}=K^{(n)}+L^{(n)}, \\
\label{ue_per_cluster_constraint_bisection_1}
&\quad\quad~(\textbf{t}^{(n)})^T\textbf{x}_{m|2}=K_{m}^{(n)},\quad m=1,2, \\
\label{bs_per_cluster_constraint_bisection}
&\quad\quad~(\textbf{1}^{(n)}-\textbf{t}^{(n)})^T\textbf{x}_{m|2} \geq \left\lceil\frac{K_{m}^{(n)}}{K_{max}}\right\rceil,\quad m=1,2, \\
\label{binary_variable_constraint_bisection}
&\quad\quad~x_{i,m}\in\{0,1\},~ i{=}1,2,{\cdots},K^{(n)}{+}L^{(n)};~m{=}1,2,
\end{align}
\end{subequations}
where $L^{(n)}$ is the number of BSs in the subnetwork represented by $\mathcal{G}^{(n)}$.
$\textbf{e}_i^{(n)}$ is an $(K^{(n)}+L^{(n)}) \times 1$ unit vector. 
$\textbf{1}^{(n)}$ is an $(K^{(n)}+L^{(n)}) \times 1$ vector of ones.
$\textbf{t}^{(n)}=[\underbrace{1,{\cdots},1}_{\text{$K^{(n)}$ 1s}},\underbrace{0,{\cdots},0}_{\text{$L^{(n)}$ 0s}}]^T$ is the auxiliary vector.
It can be seen that problem $\mathcal{P}5^{(n)}$ in (\ref{bisection_problem_vec}) has a similar form as problem $\mathcal{P}4$ in (\ref{ul_capacity_maximization_problem_vec}), which can also be solved by the branch-and-bound approach.
This algorithm is referred to as bisection clustered cell-free networking ($\text{B}\text{C}^2\text{F}$-Net). The detailed description is presented in Algorithm \ref{b2c4_networking_algorithm}.\footnote{Note that the clustered cell-free networking problem and the proposed algorithm in this paper can be straightforwardly applied in the case where UEs and BSs are equipped with multiple antennas because the proposed algorithm is solely determined by the path-loss and equipping multiple antennas at BSs and UEs has no impact on the path-loss.}
\begin{algorithm}[!t]
\caption{$\text{B}\text{C}^2\text{F}$-Net Algorithm}
\begin{algorithmic}[1]
\REQUIRE Bipartite graph $\mathcal{G}$, vertex set $\mathcal{V}$ containing all the UEs and BSs, set of UEs $\mathcal{U}$, set of BSs $\mathcal{B}$, weight matrix $\textbf{W}$, maximum allowable per-subnetwork UE number $K_{max}$
\PRINT $\mathcal{C}_1=\mathcal{V}$, $\mathcal{U}_1=\mathcal{U}$, $\mathcal{B}_1=\mathcal{B}$, $\mathcal{C}=\{\mathcal{C}_1\}$, $n=1$, $M^*=\left\lceil\frac{K}{K_{max}}\right\rceil$;
\WHILE{$n < M^*$}
\STATE Generate $\mathcal{G}^{(n)}=\mathcal{G}[\mathcal{C}_n]$;
\STATE Compute $\textbf{A}^{(n)}$ and $\textbf{L}^{(n)}$ according to (\ref{adjacency_matrix_definition}) and (\ref{L_definition}); 
\STATE $K^{(n)}=|\mathcal{U}_n|$, $L^{(n)}=|\mathcal{B}_n|$;
\STATE Compute $K_{m}^{(n)}$ for $m=1,2$ according to (\ref{kth_1_n}) and (\ref{kth_2_n});
\STATE Solve problem $\mathcal{P}5^{(n)}$ in (\ref{bisection_problem_vec}) by the branch-and-bound approach to obtain $\mathcal{C}_1^{(n)}$ and $\mathcal{C}_2^{(n)}$;
\STATE $\mathcal{C}=\mathcal{C}\setminus\{\mathcal{C}_n\}$;
\STATE $\mathcal{C}=\mathcal{C}\cup\{\mathcal{C}_1^{(n)},\mathcal{C}_2^{(n)}\}$;
\STATE Renumber the subnetworks in $\mathcal{C}$ in ascending order of the number of UEs in each subnetwork;
\STATE $n=n+1$;
\ENDWHILE
\ENSURE $\mathcal{C}=\{\mathcal{C}_{1}, \mathcal{C}_{2}, \cdots, \mathcal{C}_{M^*}\}$
\end{algorithmic}
\label{b2c4_networking_algorithm}
\end{algorithm}

\vspace{-4mm}
\subsection{Complexity Analysis}
\label{Low_Complexity_Algorithm_Complexity_Analysis}
Note that the proposed $\text{B}\text{C}^2\text{F}$-Net algorithm solves the formulated clustered cell-free networking problem $\mathcal{P}4$ in (\ref{ul_capacity_maximization_problem_vec}) by iteratively bisecting the subnetwork containing the largest number of UEs using the branch-and-bound approach until $M^{*}$ subnetworks are produced. 
Among $M^*-1$ bisections, the first bisection deals with the network with the largest size, i.e., $K$ users and $L$ BSs, and hence has the largest number of binary variables to determine.
As the subnetwork size is roughly halved after each bisection and the computational complexity decreases exponentially as the number of binary variables decreases according to Section \ref{Complexity_Analysis}, it can be concluded that the first bisection dominates the total computational complexity of the $\text{B}\text{C}^2\text{F}$-Net algorithm.
According to (\ref{first_iteration_complexity}), the per-iteration computational complexity of the first bisection is obtained as $\mathcal{O}\left(\ln\left(\frac{1}{\epsilon}\right)(K+L)^{3.5}\right)$.
As the total computational complexity is dominated by the complexity of the first bisection, our proposed $\text{B}\text{C}^2\text{F}$-Net algorithm has the computational complexity of
\begin{align}
\label{bc2f_net_complexity_ub}
\mathcal{O}\left(N_{itr}^{(1)}\ln\left(\frac{1}{\epsilon}\right)(K+L)^{3.5}\right),
\end{align}
where $N_{itr}^{(1)}$ is the number of iterations in the first bisection.

It can be seen that in contrast to the complexity of the original branch-and-bound based solution that increases with the optimal number of subnetworks $M^{*}=\lceil{K/K_{max}}\rceil$ in the order of $\mathcal{O}\left(\left\lceil K/K_{max}\right\rceil^{3.5}\right)$, the complexity of the proposed $\text{B}\text{C}^2\text{F}$-Net algorithm remains the same as $\lceil{K/K_{max}}\rceil$ increases.
Moreover, since the number of binary entries in decision matrix $\textbf{X}^{(1)}$ is usually much smaller than that in the decision matrix $\textbf{X}$ of problem $\mathcal{P}4$, the number of iterations $N_{itr}^{(1)}$ in (\ref{bc2f_net_complexity_ub}) is much smaller than $N_{itr}$ in (\ref{total_complexity_ub}) \cite{1_bit_mmimo_precoding_partial_bb}.
For instance, in the worst case the number of iterations in the original branch-and-bound based solution is $N_{itr}=2^{(K+L)\lceil{K/K_{max}}\rceil}$ as discussed in Section \ref{Complexity_Analysis}, while the number of iterations $N_{itr}^{(1)}$ with the proposed $\text{B}\text{C}^2\text{F}$-Net algorithm is $N_{itr}^{(1)}=2^{2(K+L)}$, much smaller than $N_{itr}$ when $M^{*}=\lceil{K/K_{max}}\rceil>2$.
Therefore, it can be concluded from (\ref{total_complexity_ub}) and (\ref{bc2f_net_complexity_ub}) that the computational complexity of the proposed $\text{B}\text{C}^2\text{F}$-Net algorithm is much lower than that of directly solving problem $\mathcal{P}4$ by the branch-and-bound method. 

\begin{figure}[!t]
\centering
\subfloat[$K=20$]{
\includegraphics[width=0.15\textwidth]{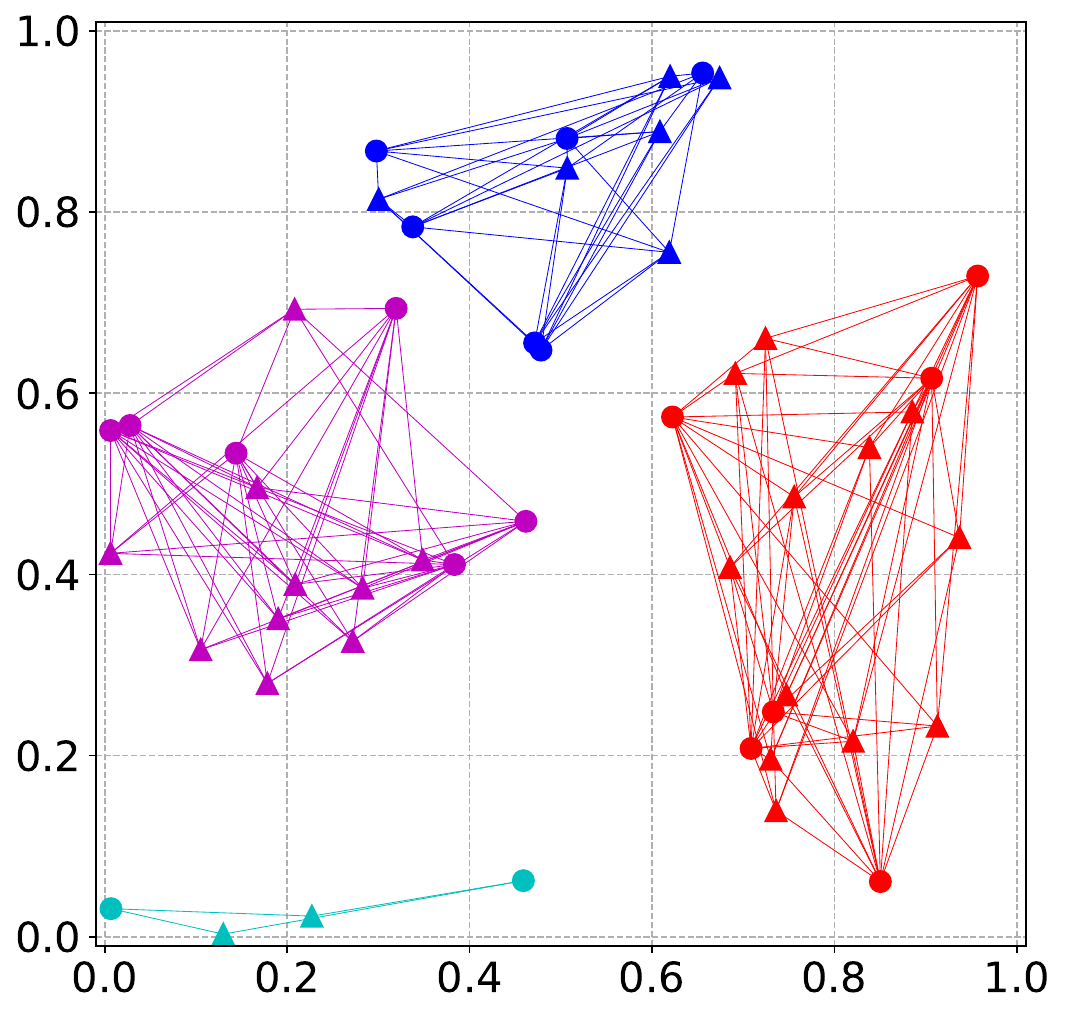}
\label{fast_min_sum_quad_30_20}}
\subfloat[$K=30$]{
\includegraphics[width=0.15\textwidth]{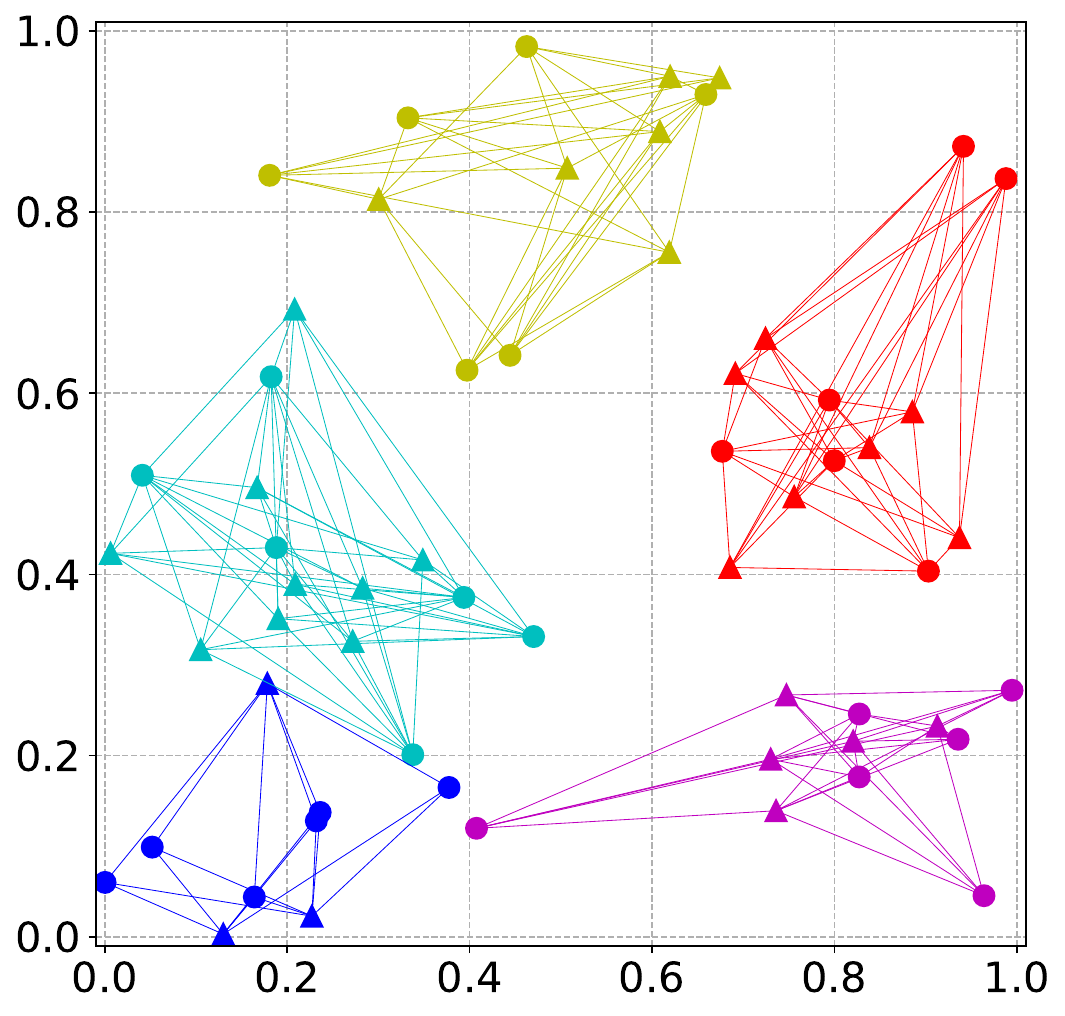}
\label{fast_min_sum_quad_30_30}}
\subfloat[$K=40$]{
\includegraphics[width=0.15\textwidth]{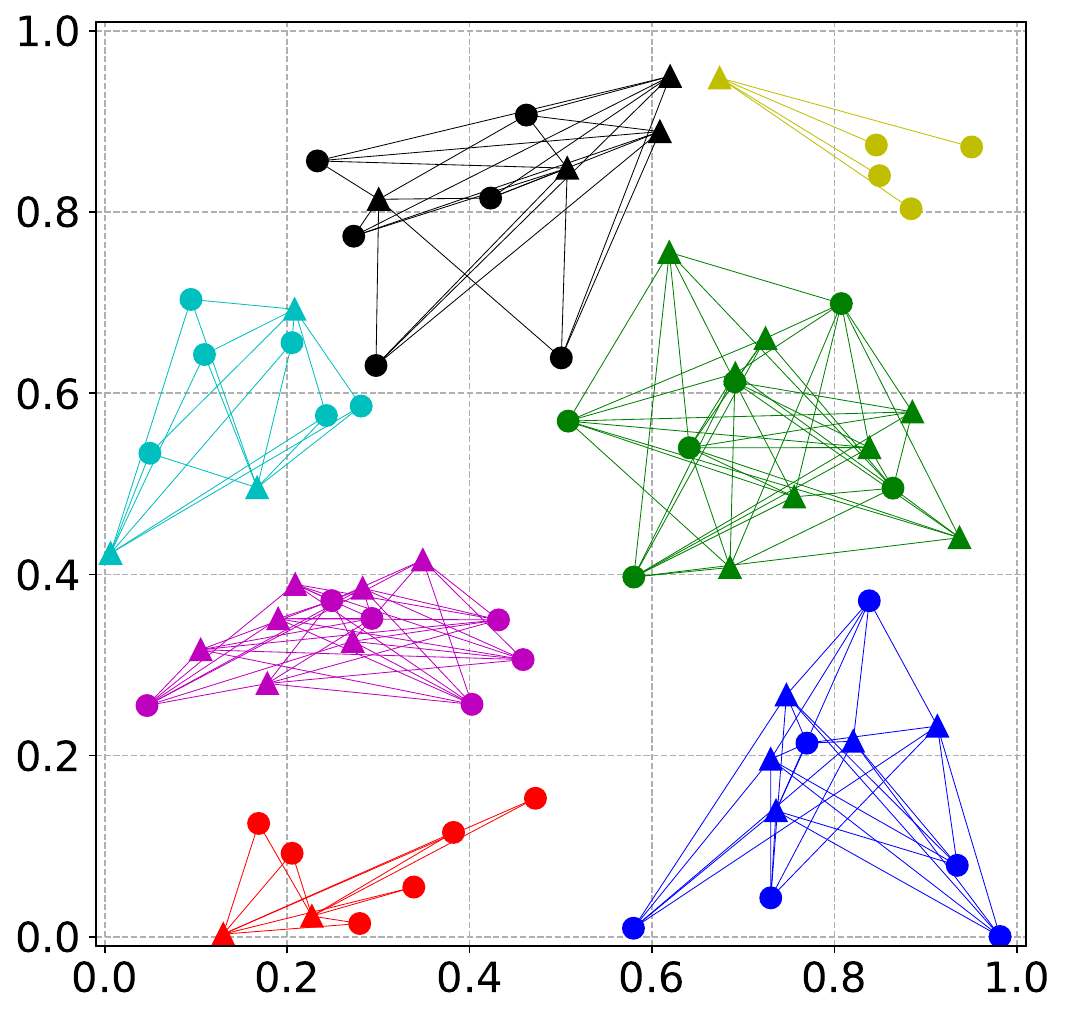}
\label{fast_min_sum_quad_30_40}}
\caption{Snapshots of the network decomposition results with the branch-and-bound based clustered cell-free networking scheme. Both the BS layout and UE layout are randomly generated. UEs and BSs are represented by circles and triangles, respectively. UEs and BSs belonging to the same subnetwork along with the lines connecting them are in the same color. $K_{max}=6$, $L=30$, $\alpha=4$, $P/N_0=10$dB.}
\label{fast_min_sum_quad_30_20_30_40}
\end{figure}

\section{Simulation Results}
\label{Simulation_Results}
This section presents the simulation results to demonstrate the effectiveness and performance of the branch-and-bound based solution presented in Section \ref{Proposed_Solution} and the $\text{B}\text{C}^2\text{F}$-Net algorithm proposed in Section \ref{Suboptimal_Low_Complexity_Algorithm}.
A large-scale wireless network with $K$ UEs and $L$ BSs randomly distributed within a square area is considered, where the locations of UEs and BSs follow an independent and identical uniform distribution.
The average sum ergodic capacity and the average running time are employed to evaluate the performance and complexity of the algorithms, respectively, which are obtained by averaging over 200 random realizations of UE layouts and BS layouts.

\subsection{Effectiveness of Branch-and-Bound Based Scheme}
\label{Effectiveness_of_Lower_Bound_Approximation}

Fig. \ref{fast_min_sum_quad_30_20_30_40} presents the network decomposition results with the branch-and-bound based clustered cell-free networking scheme, where the number of BSs is fixed at $L = 30$, and the number of UEs $K$ is 20, 30 and 40 to simulate three different scenarios with $K < L$, $K = L$, and $K > L$, respectively.
It is shown in Fig. \ref{fast_min_sum_quad_30_20_30_40} that the branch-and-bound method is capable of decomposing a network into non-overlapping subnetworks while simultaneously guaranteeing that the number of UEs in each subnetwork is no larger than the preset limit $K_{max}$ in all three cases, which verifies its effectiveness.
Moreover, it can be seen that there are always $M^{*}-1$ out of $M^{*}$ subnetworks containing exactly $K_{max}$ UEs after network decomposition, which indicates that the decomposed subnetworks are always balanced.
This is because with the aim of maximizing the sum ergodic capacity, a subnetwork should contain as many UEs as possible, i.e., $K_{max}$ UEs.
For the same reason, there could be one small subnetwork when the total number of UEs $K$ is slightly larger than an integer multiple of $K_{max}$, as can be observed in Fig. \ref{fast_min_sum_quad_30_20_30_40}\subref{fast_min_sum_quad_30_20}.

\begin{figure}[!t]
\centering
\subfloat[$K_{max}=10$]{
\includegraphics[width=0.23\textwidth]{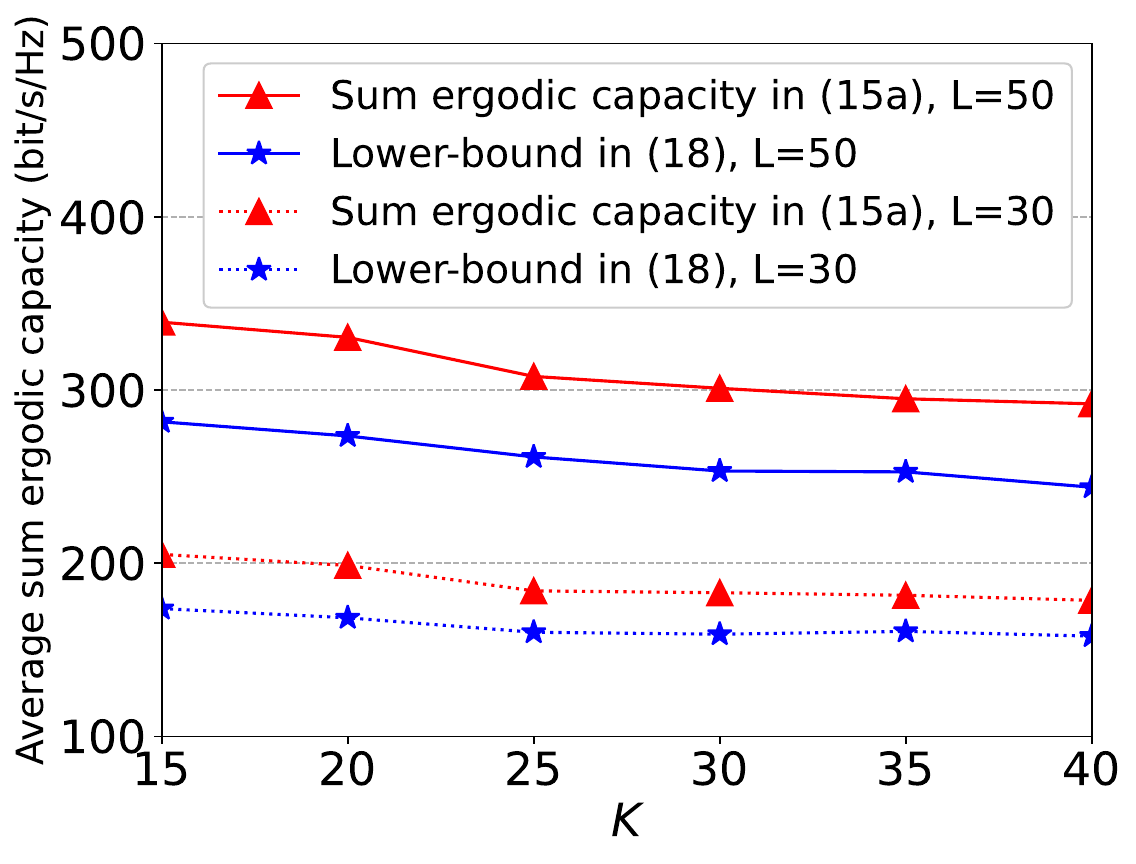}
\label{exact_capacity_vs_lower_bound_various_ue_num}}
\subfloat[$K=30$]{
\includegraphics[width=0.23\textwidth]{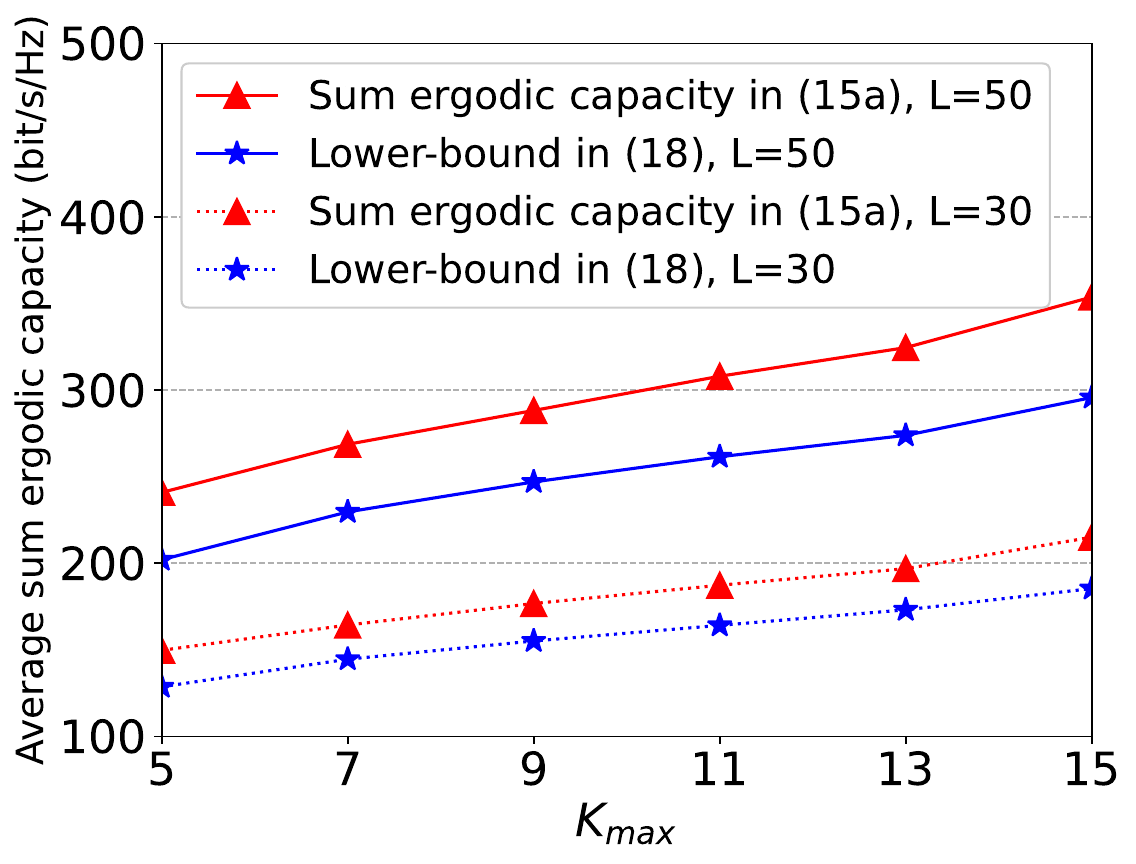}
\label{exact_capacity_vs_lower_bound_various_threshold}}
\caption{Average sum ergodic capacity in (\ref{ul_capacity_objctive_with_M*}) and its lower-bound in (\ref{lower_bound_objective}) with the branch-and-bound based clustered cell-free networking scheme. $\alpha=4$, $P/N_0=10$dB, $L=30, 50$.}
\label{lb_comparison}
\end{figure}

Note that the branch-and-bound based clustered cell-free networking scheme is developed by replacing the sum ergodic capacity in (\ref{ul_capacity_objctive_with_M*}) with its lower-bound in (\ref{lower_bound_objective}) to transform the original NP-hard clustered cell-free networking problem into an integer convex programming problem.
Now let us demonstrate the rationale of the adopted lower-bound approximation.
Considering two scenarios with $L=30$ and $L=50$, Fig. \ref{lb_comparison} presents the achieved average sum ergodic capacity after decomposing the network by the branch-and-bound based method and its lower-bound in  (\ref{lower_bound_objective}).
It can be seen from both Fig. \ref{lb_comparison}\subref{exact_capacity_vs_lower_bound_various_ue_num} and Fig. \ref{lb_comparison}\subref{exact_capacity_vs_lower_bound_various_threshold} that the gap between the average sum ergodic capacity in (\ref{ul_capacity_objctive_with_M*}) and its lower-bound in (\ref{lower_bound_objective}) is almost a constant by increasing either the number of users $K$ or the maximum allowable per-subnetwork UE number $K_{max}$, indicating that maximizing the lower-bound of the sum ergodic capacity instead of the original sum ergodic capacity has little impact on devising the clustered cell-free networking scheme.
Next, let us demonstrate the effectiveness of the proposed $\text{B}\text{C}^2\text{F}$-Net Algorithm, which has much lower computational complexity than the branch-and-bound based scheme.

\begin{figure}[!t]
\centering
\subfloat[$K=20$]{
\includegraphics[width=0.15\textwidth]{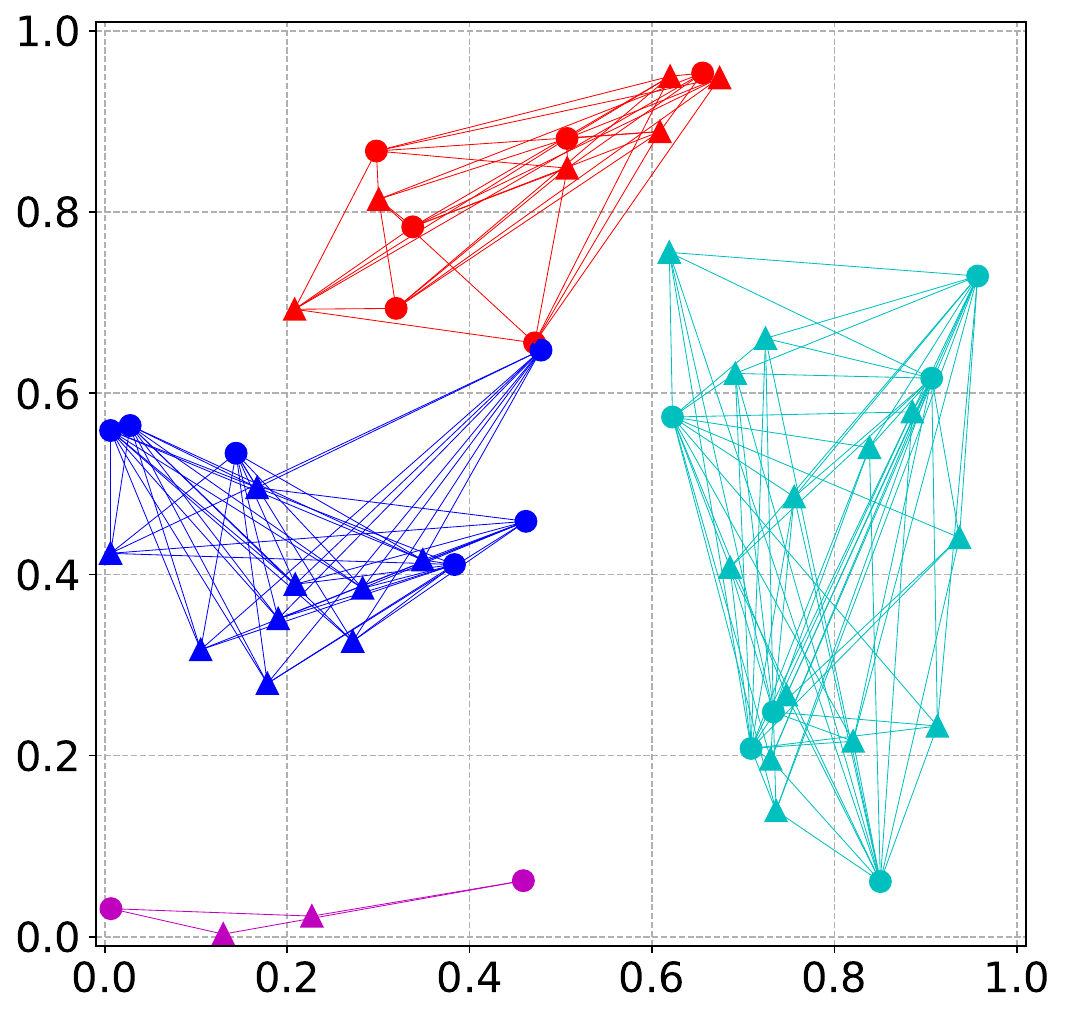}
\label{hierarchical_normal_30_20}}
\subfloat[$K=30$]{
\includegraphics[width=0.15\textwidth]{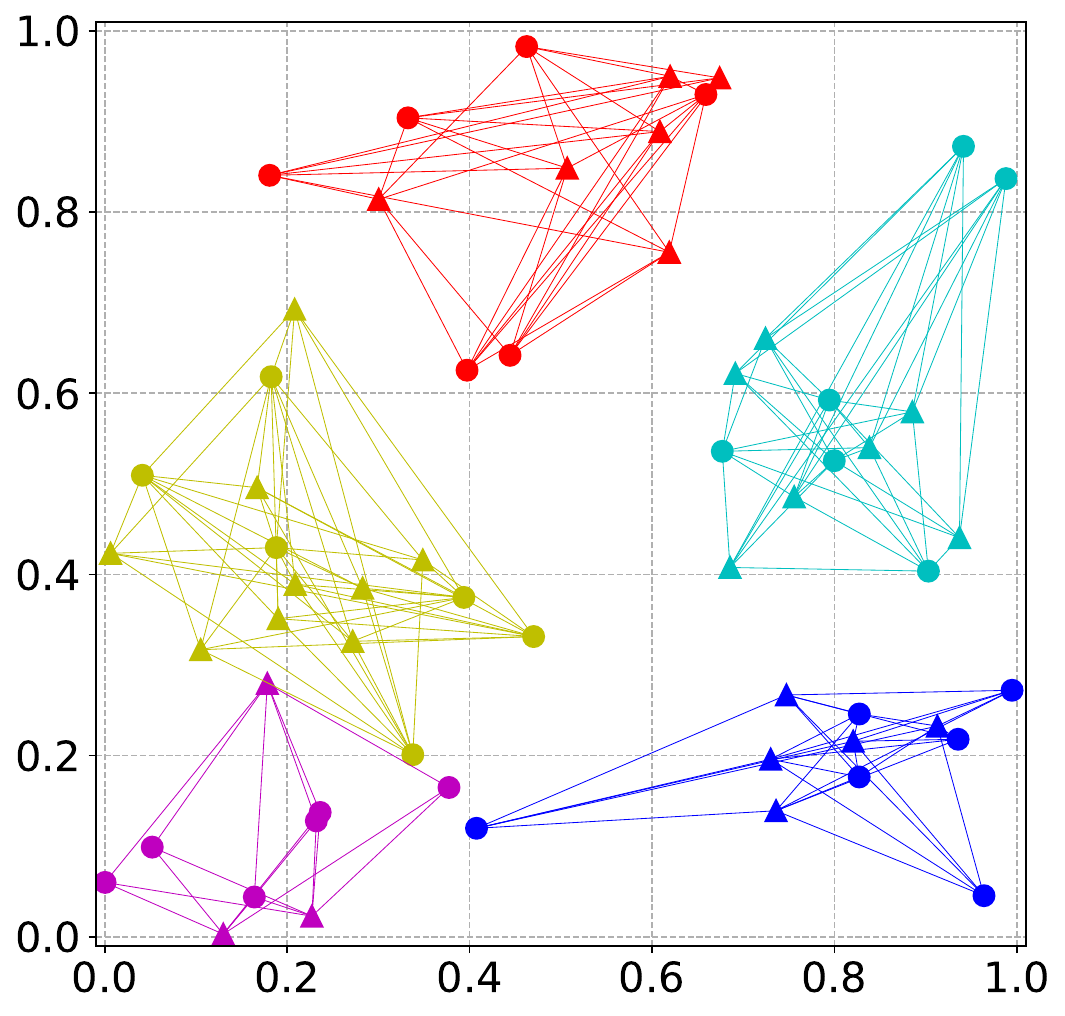}
\label{hierarchical_normal_30_30}}
\subfloat[$K=40$]{
\includegraphics[width=0.15\textwidth]{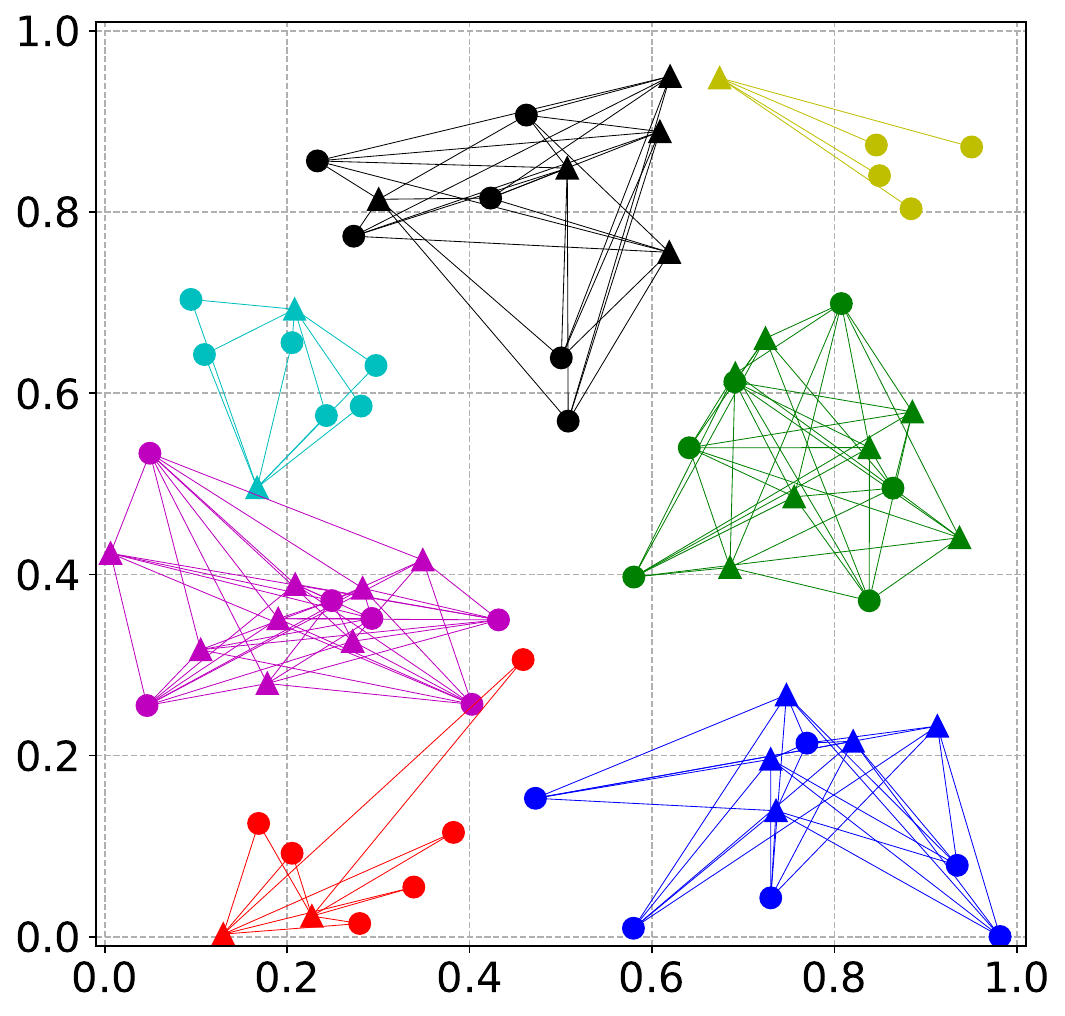}
\label{hierarchical_normal_30_40}}
\caption{Snapshots of the network decomposition results with the proposed $\text{B}\text{C}^2\text{F}$-Net algorithm. Both the BS layout and UE layout are the same as in Fig. 2. UEs and BSs are represented by circles and triangles, respectively. UEs and BSs belonging to the same subnetwork along with the lines connecting them are in the same color. $K_{max}=6$, $L=30$, $\alpha=4$, $P/N_0=10$dB.}
\label{hierarchical_normal_30_20_30_40}
\end{figure}

\begin{figure}[!t]
\centering
\subfloat[$K_{max}=10$]{
\includegraphics[width=0.23\textwidth]{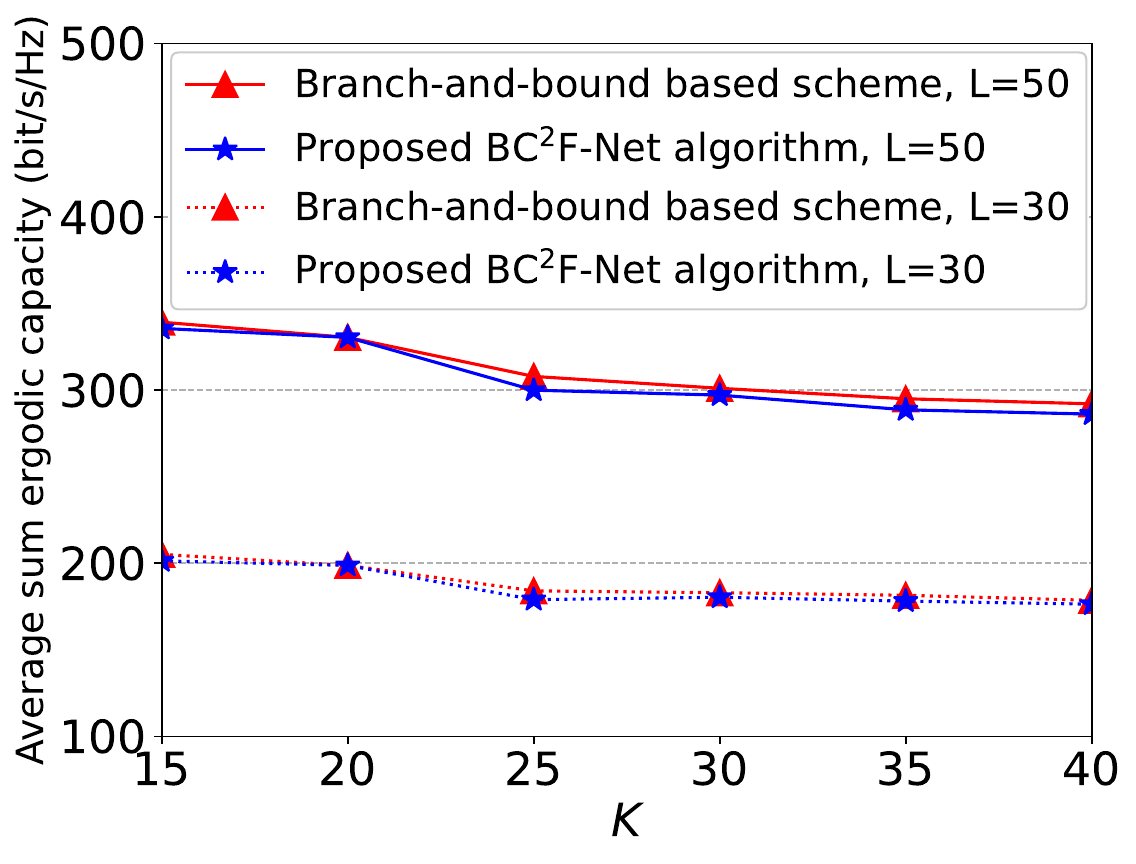}
\label{ul_capacity_comparison_various_ue_num}}
\subfloat[$K=30$]{
\includegraphics[width=0.23\textwidth]{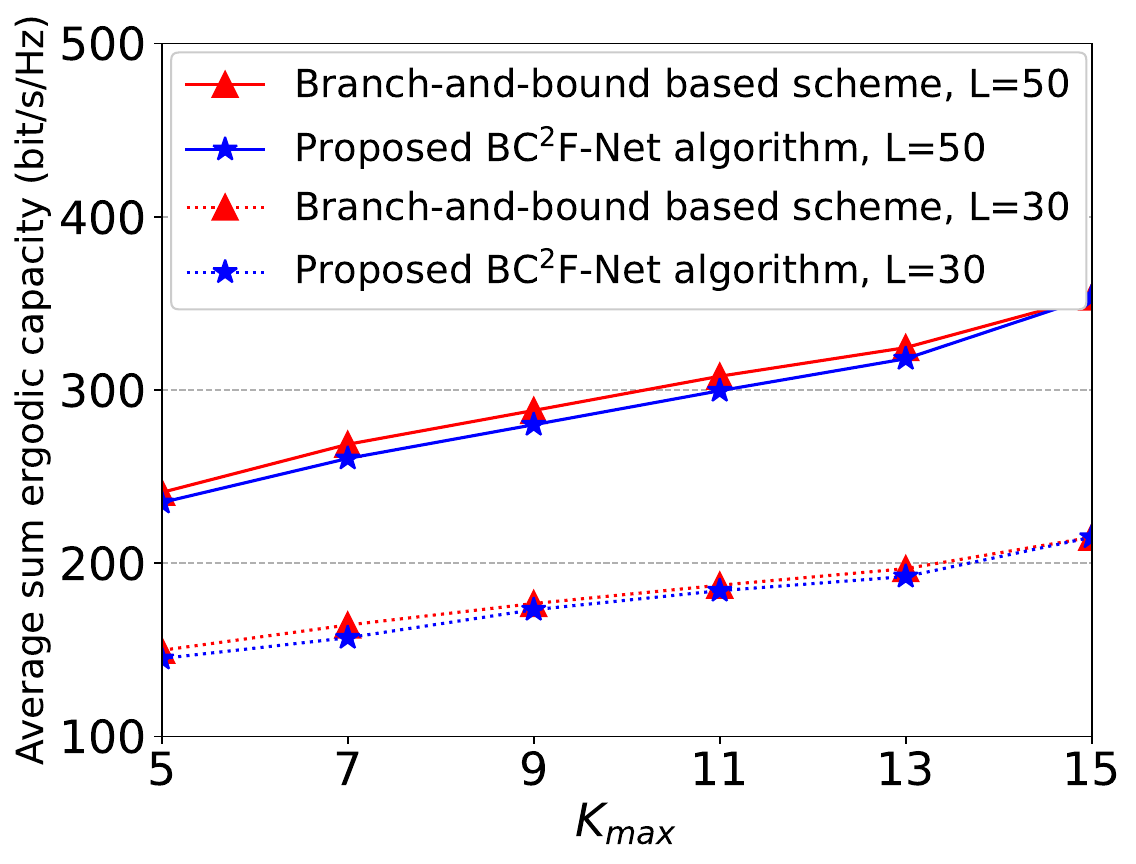}
\label{ul_capacity_comparison_various_kmax}}
\caption{Average sum ergodic capacity with the proposed $\text{B}\text{C}^2\text{F}$-Net algorithm and the branch-and-bound based scheme. $\alpha=4$, $P/N_0=10$dB, $L=30, 50$.}
\label{capacity_comparison_fast_min_sum_quad_vs_hierarchical_normal}
\end{figure}

\begin{figure}[!t]
\centering
\subfloat[$K_{max}=10$]{
\includegraphics[width=0.23\textwidth]{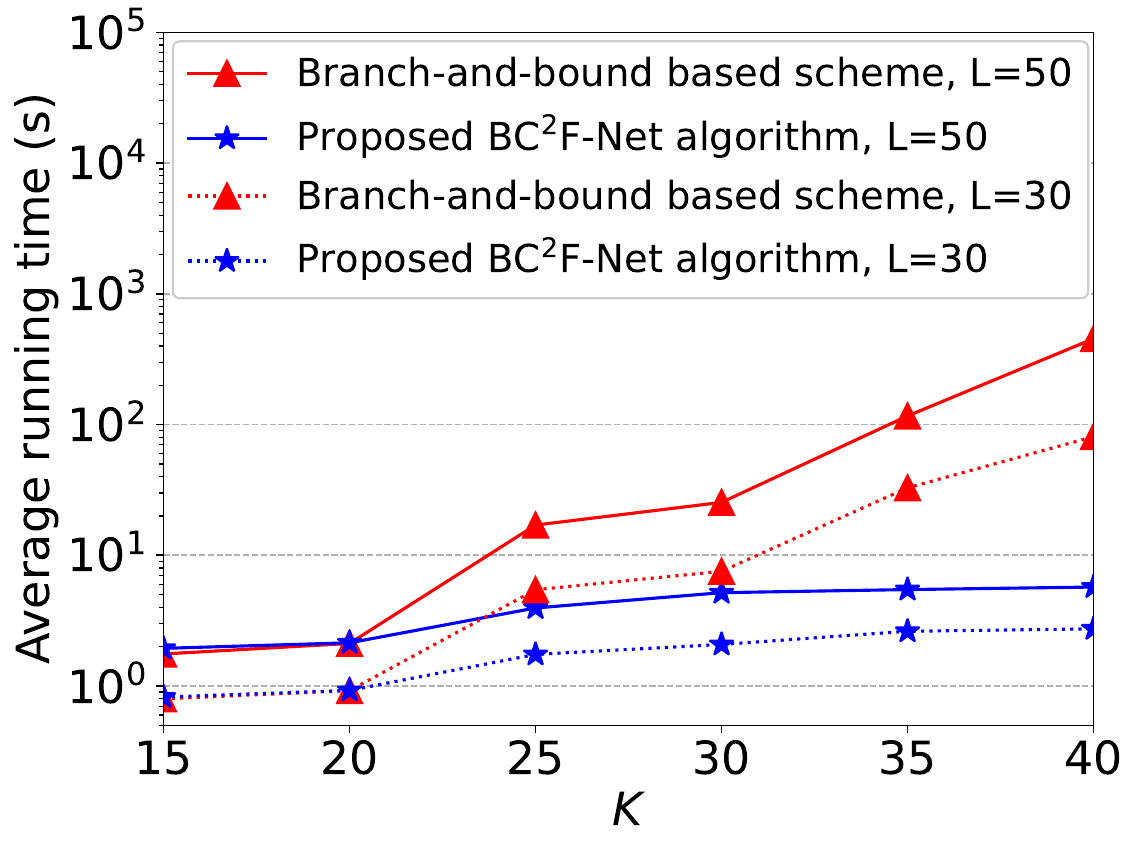}
\label{solve_time_comparison_various_ue_num}}
\subfloat[$K=30$]{
\includegraphics[width=0.23\textwidth]{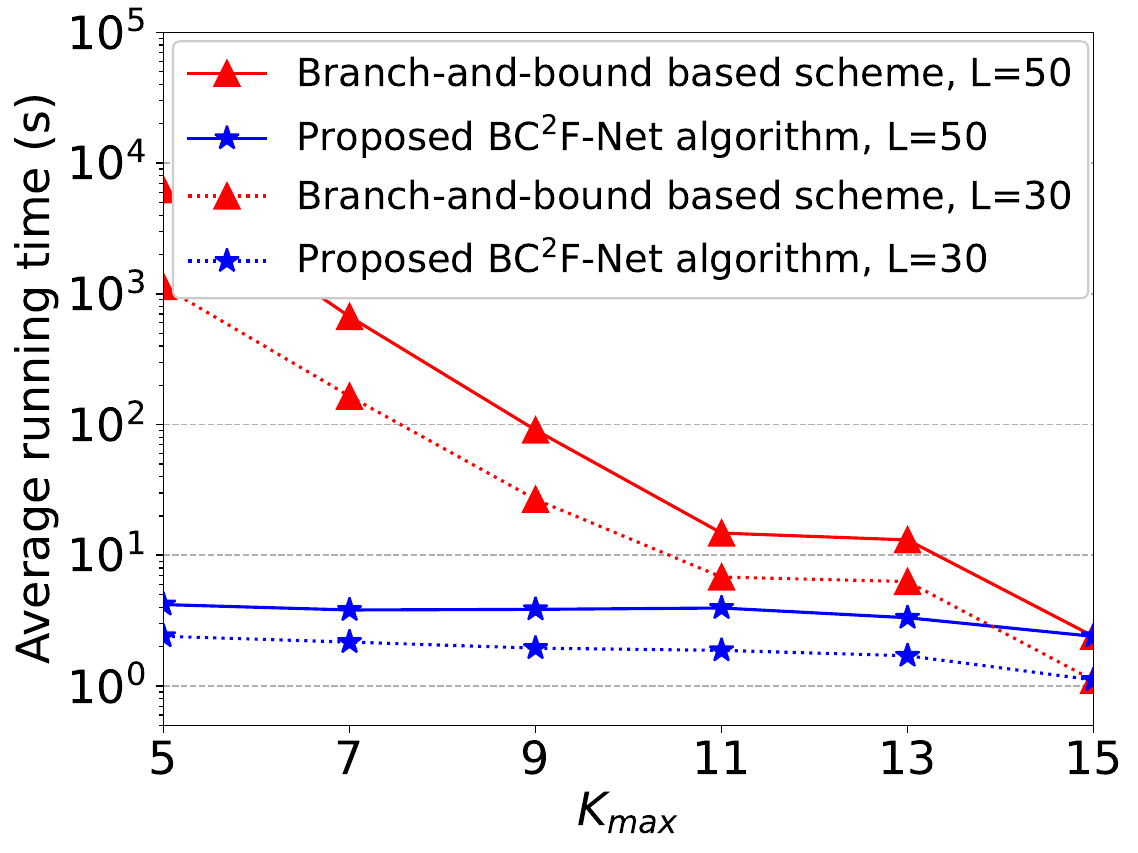}
\label{solve_time_comparison_various_kmax}}
\caption{Average running time of the proposed  $\text{B}\text{C}^2\text{F}$-Net algorithm and the branch-and-bound based scheme. $\alpha=4$, $P/N_0=10$dB, $L=30, 50$.}
\label{running_time_comparison_fast_min_sum_quad_vs_hierarchical_normal}
\end{figure}

\subsection{Effectiveness of Proposed $\text{B}\text{C}^2\text{F}$-Net Algorithm}
\label{Effectiveness_of_B2C4_Networking_Algorithm}
Fig. \ref{hierarchical_normal_30_20_30_40} presents the network decomposition results obtained from the same network layouts shown in Fig. \ref{fast_min_sum_quad_30_20_30_40} by adopting the proposed $\text{B}\text{C}^2\text{F}$-Net algorithm.
By comparing Fig. \ref{fast_min_sum_quad_30_20_30_40} and Fig. \ref{hierarchical_normal_30_20_30_40}, it can be observed that the proposed $\text{B}\text{C}^2\text{F}$-Net algorithm is capable of producing similar network decomposition results as the branch-and-bound based clustered cell-free networking scheme in all three different scenarios yet with much lower computational complexity.

To evaluate the performance of the proposed $\text{B}\text{C}^2\text{F}$-Net algorithm, Fig. \ref{capacity_comparison_fast_min_sum_quad_vs_hierarchical_normal} compares the average sum ergodic capacity of the proposed $\text{B}\text{C}^2\text{F}$-Net algorithm and the branch-and-bound based scheme when the number of BSs $L$ is fixed at 30 and 50.
It can be found from Fig. \ref{capacity_comparison_fast_min_sum_quad_vs_hierarchical_normal} that the proposed $\text{B}\text{C}^2\text{F}$-Net algorithm brings at most 4.2\% degradation in sum capacity compared to the branch-and-bound based scheme.
Moreover, the low-complexity $\text{B}\text{C}^2\text{F}$-Net algorithm achieves the same average sum ergodic capacity as the branch-and-bound based solution, when the number of UEs $K=20$ in Fig. \ref{capacity_comparison_fast_min_sum_quad_vs_hierarchical_normal}\subref{ul_capacity_comparison_various_ue_num} and the maximum allowable per-subnetwork UE number $K_{max}=15$ in Fig. \ref{capacity_comparison_fast_min_sum_quad_vs_hierarchical_normal}\subref{ul_capacity_comparison_various_kmax}.
This is because in both cases, the network is decomposed into two subnetworks, implying that the low-complexity $\text{B}\text{C}^2\text{F}$-Net algorithm that bisects the network hierarchically brings no performance loss when the number of UEs in the networks is an integer multiple of maximum allowable per-subnetwork UE number $K_{max}$.

Fig. \ref{running_time_comparison_fast_min_sum_quad_vs_hierarchical_normal} shows the running time of the proposed $\text{B}\text{C}^2\text{F}$-Net algorithm and the branch-and-bound based scheme under the same settings as Fig. \ref{capacity_comparison_fast_min_sum_quad_vs_hierarchical_normal}.
We can see from the figure that the average running time of the $\text{B}\text{C}^2\text{F}$-Net algorithm is much shorter compared to the branch-and-bound based scheme.
For instance, when the number of UEs $K=30$, the maximum allowable per-subnetwork UE number $K_{max}=5$ and the number of BSs $L=50$, the running time of our $\text{B}\text{C}^2\text{F}$-Net algorithm is only 0.06\% of that of the branch-and-bound based method.
Moreover, the average running time of the $\text{B}\text{C}^2\text{F}$-Net algorithm increases at a much slower pace as the number of UEs $K$ increases or as the maximum allowable per-subnetwork UE number $K_{max}$ decreases.
For instance, it is shown in Fig. \ref{running_time_comparison_fast_min_sum_quad_vs_hierarchical_normal}\subref{solve_time_comparison_various_ue_num} that when the number of UEs $K$ increases from 20 to 40 in the case with 50 BSs, the running time of the branch-and-bound based scheme increases from around 2s to 454s while the running time of our $\text{B}\text{C}^2\text{F}$-Net algorithm only increases from 2s to 5.7s, corroborating that the bisection approach proposed to partition the network hierarchically can effectively flatten the computational complexity increase.   
Altogether, Fig. \ref{capacity_comparison_fast_min_sum_quad_vs_hierarchical_normal} and Fig. \ref{running_time_comparison_fast_min_sum_quad_vs_hierarchical_normal} validate that our proposed $\text{B}\text{C}^2\text{F}$-Net algorithm is capable of significantly reducing the computational complexity with negligible degradation of sum ergodic capacity.

\subsection{Comparison with Brute-Force Search}
\label{Comparison_with_Brute_Force_Search}
Fig. \ref{brute_force_vs_bc2fnet} presents the sum ergodic capacity with both the proposed $\text{B}\text{C}^2\text{F}$-Net algorithm and the optimal brute-force search that finds the optimal network decomposition by traversing all the $M^{K+L}$ possibilities, under 30 random realizations of UE and BS layouts when the number of BSs $L=6$ and the number of UEs $K=5$ and $8$.
It can be seen from Fig. \ref{brute_force_vs_bc2fnet} that the proposed $\text{B}\text{C}^2\text{F}$-Net algorithm achieves almost the same sum ergodic capacity as the optimal brute-force search in most cases.
This verifies the near-optimal performance of our proposed $\text{B}\text{C}^2\text{F}$-Net algorithm which has a much lower computational complexity. 

\begin{figure}[!t]
\centering
\subfloat[$K=5$]{
\includegraphics[width=0.23\textwidth]{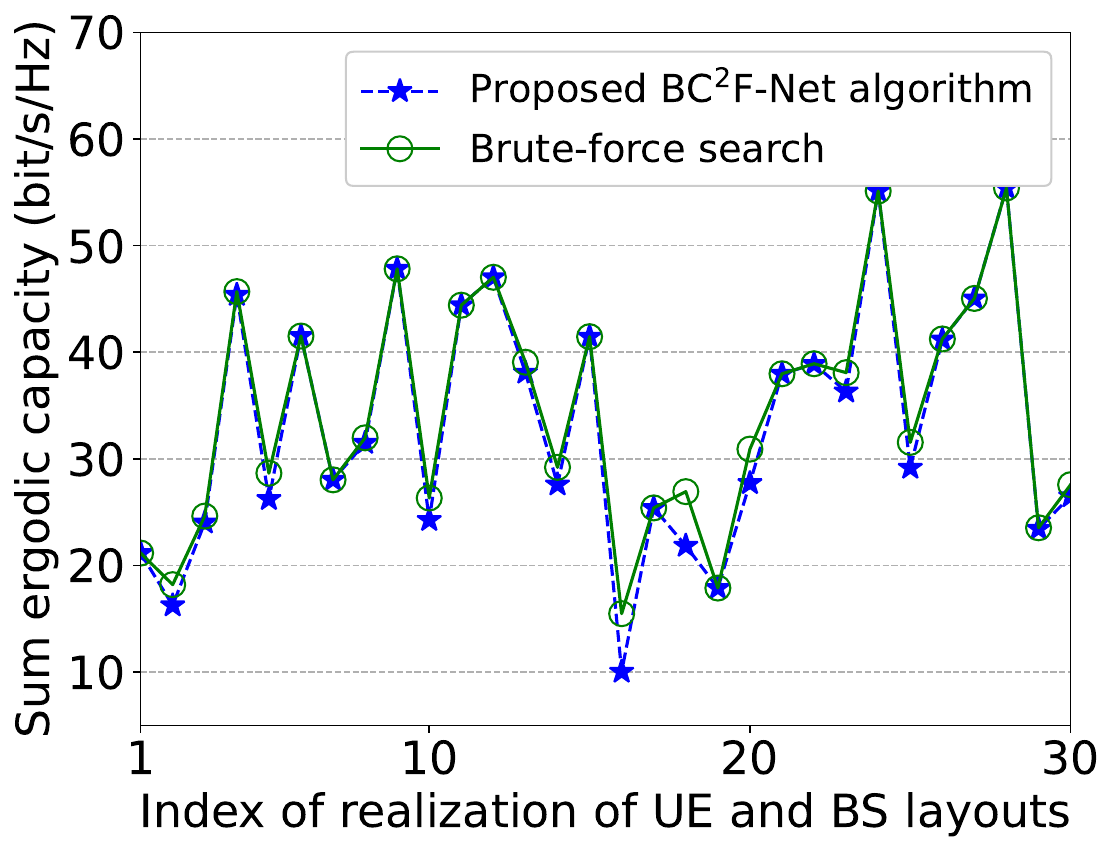}
\label{brute_force_vs_bc2fnet_k_5}}
\subfloat[$K=8$]{
\includegraphics[width=0.23\textwidth]{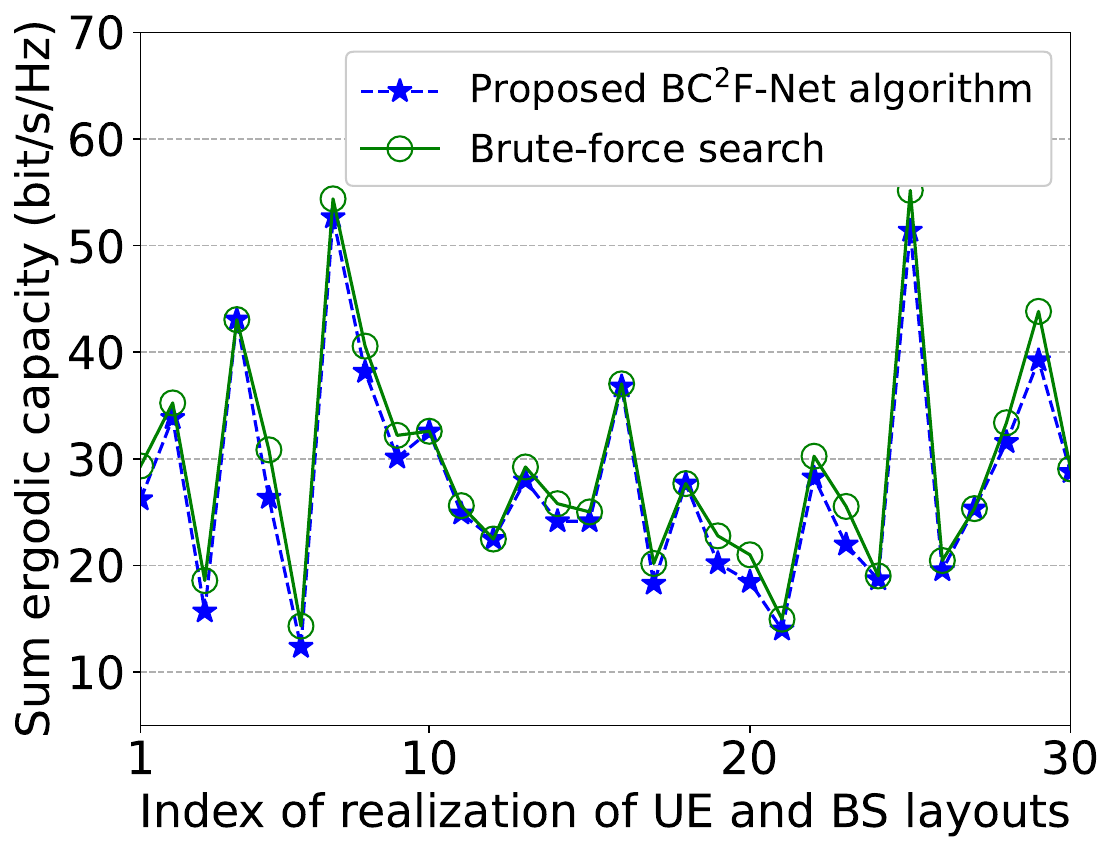}
\label{brute_force_vs_bc2fnet_k_8}}
\caption{Sum ergodic capacity under 30 realizations of UE and BS layouts with both the optimal brute-force search and the proposed $\text{B}\text{C}^2\text{F}$-Net algorithm. $\alpha=4$, $P/N_0=10$dB, $L=6$, $K_{max}=3$.}
\label{brute_force_vs_bc2fnet}
\end{figure}

\begin{figure*}[!t]
\centering
\subfloat[$K_{max}=10$]{
\includegraphics[width=0.445\textwidth]{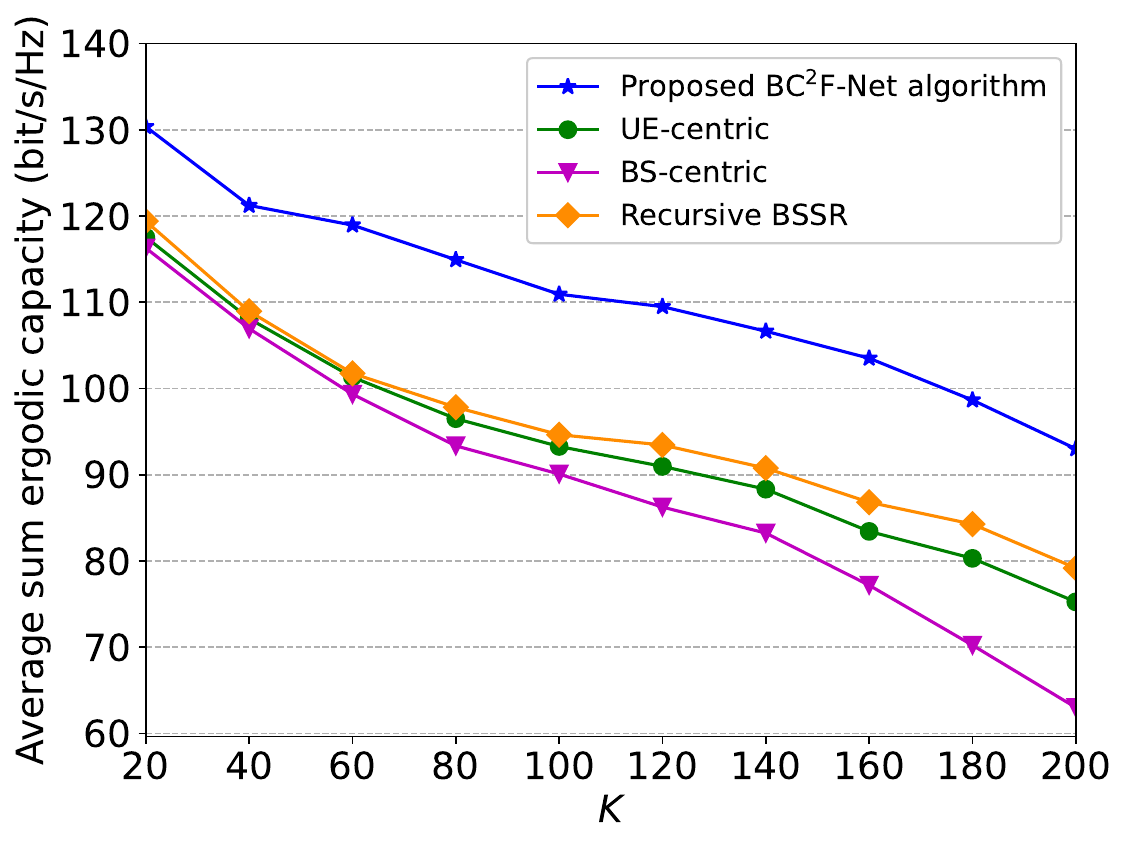}
\label{ul_capacity_vs_ue_num_kth_6}}
\hspace{4mm}
\subfloat[$K=100$]{
\includegraphics[width=0.44\textwidth]{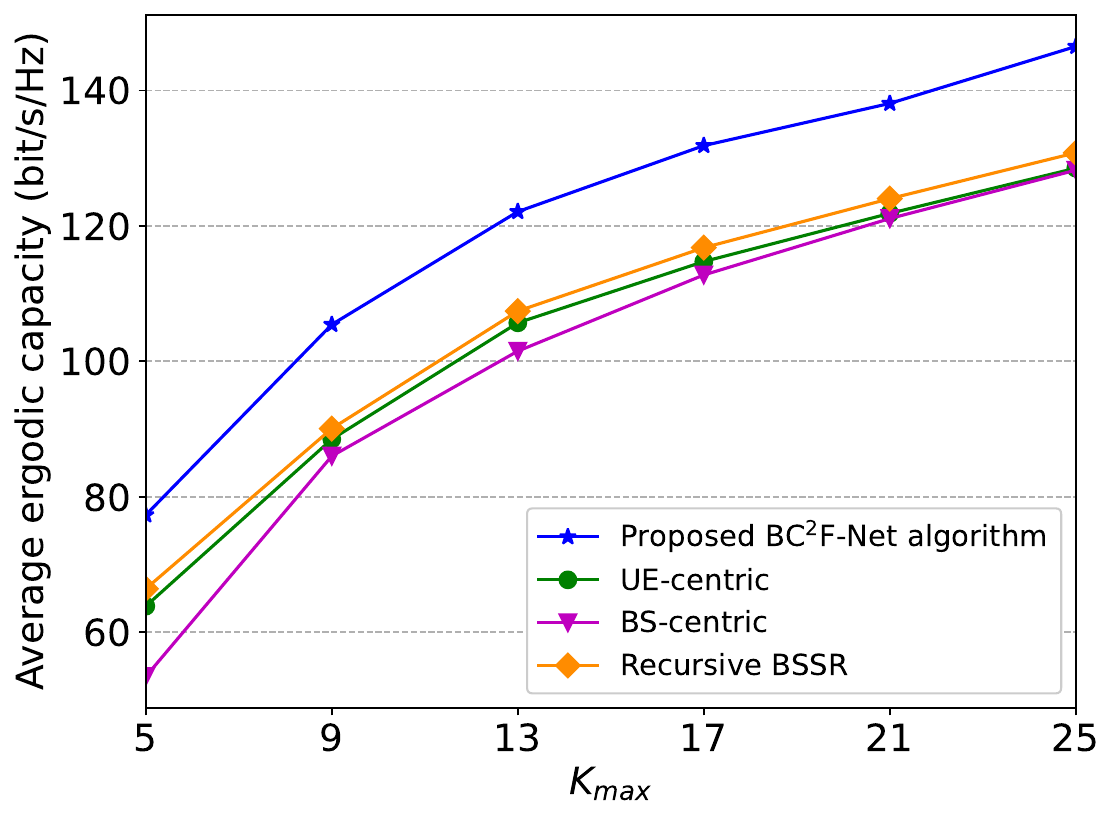}
\label{ul_capacity_vs_threshold_30_30}}
\caption{Average sum ergodic capacity with the proposed $\text{B}\text{C}^2\text{F}$-Net algorithm and the state-of-the-art benchmarks. $\alpha=4$, $P/N_0=10$dB, $L=20$.}
\label{average_sum_ergodic_capacity}
\end{figure*}

\subsection{Comparison with Benchmarks}
\label{Average_Sum_Capacity}
In this subsection, the proposed $\text{B}\text{C}^2\text{F}$-Net algorithm is compared with the state-of-the-art benchmarks to demonstrate its performance.
The details of the state-of-the-art benchmarks chosen in this paper are summarized as follows:
\begin{itemize}
\item \textbf{UE-Centric} \cite{clustered_cf_mmimo}: The UE-centric benchmark is the two-stage clustering algorithm proposed in \cite{clustered_cf_mmimo}, which recursively applies K-means algorithm to cluster UEs while ensuring that the number of UEs in each cluster is no larger than the maximum allowable per-subnetwork UE number $K_{max}$ first, and then assign BSs to UE clusters to form subnetworks.  
\item \textbf{BS-Centric} \cite{clustered_cf_mmimo}: Since most existing two-stage BS-centric clustering algorithms placed no constraint on the number of UEs in each subnetwork, for the sake of fair comparison, we adopt the same recursive clustering approach in \cite{clustered_cf_mmimo} while first clustering BSs and then assigned UEs to BS clusters to form subnetworks.
\item \textbf{Recursive BSSR} \cite{optimal_decomposition_networks}: As a representative one-stage clustered cell-free networking algorithm, the BSSR algorithm in \cite{optimal_decomposition_networks} optimally decomposes the network based on a bipartite graph representation of the network.
As \cite{optimal_decomposition_networks} did not control the number of UEs in each subnetwork neither, the same recursive mechanism adopted in the above two-stage benchmarks is adopted on top of the BSSR algorithm by recursively executing it to redecompose the subnetworks that have more than $K_{max}$ UEs until the per-subnetwork UE number constraint is fulfilled.
\end{itemize}

Fig. \ref{average_sum_ergodic_capacity} presents the average sum ergodic capacity with the proposed $\text{B}\text{C}^2\text{F}$-Net algorithm and the state-of-the-art benchmarks by increasing the number of UEs $K$ or the maximum allowable per-subnetwork UE number $K_{max}$.
It can be seen from Fig. \ref{average_sum_ergodic_capacity}\subref{ul_capacity_vs_ue_num_kth_6} that as the number of UEs $K$ increases, the average sum ergodic capacity decreases.
This is due to the fact that with a larger number of UEs, there would be more subnetworks after decomposition as each subnetwork can contain at most $K_{max}$ UEs, which leads to higher inter-subnetwork interference and thus degrades sum capacity.
For the same reason, it can be observed in Fig. \ref{average_sum_ergodic_capacity}\subref{ul_capacity_vs_threshold_30_30} that the average sum ergodic capacity increases as the maximum allowable per-subnetwork UE number $K_{max}$ increases thanks to the reduced inter-subnetwork interference and the higher joint processing gain achieved in each subnetwork. 

Importantly, we can clearly see from both Fig. \ref{average_sum_ergodic_capacity}\subref{ul_capacity_vs_ue_num_kth_6} and Fig. \ref{average_sum_ergodic_capacity}\subref{ul_capacity_vs_threshold_30_30} that significant gains in the average sum ergodic capacity can be achieved by our proposed $\text{B}\text{C}^2\text{F}$-Net algorithm over the benchmarks. Specifically, our proposed $\text{B}\text{C}^2\text{F}$-Net algorithm achieves a maximum performance gain of {25}\% and {20}\% compared to the two-stage BS-centric and UE-centric benchmarks, respectively.
When comparing with the recursive BSSR benchmark \cite{optimal_decomposition_networks} that has the highest capacity among all the benchmarks, our $\text{B}\text{C}^2\text{F}$-Net algorithm still achieves around {16}\% performance gain. 
The performance gain originates from the fact that in addition to optimizing the network decomposition by taking into account the information of both UEs and BSs, the proposed $\text{B}\text{C}^2\text{F}$-Net algorithm produces the optimal number of subnetworks for given per-subnetwork UE number constraint.
The superior performance of our $\text{B}\text{C}^2\text{F}$-Net algorithm highlights the importance of the joint optimization of the number of subnetworks and the corresponding network decomposition.

\section{Conclusion}
\label{Conclusion}
This paper studied the clustered cell-free networking problem under the practical joint processing constraint that the number of UEs in each subnetwork is bounded by a given limit.
With the objective of maximizing the sum ergodic capacity of the network, the optimal number of subnetworks and the corresponding network decomposition were jointly optimized.
The optimal number of subnetworks was first obtained as a simple closed-form expression thanks to the monotone property of the objective function.
The network decomposition was then devised by successfully transforming the clustered cell-free networking problem into an integer convex programming problem and solving it by the branch-and-bound method. 
Due to the high computational complexity of the branch-and-bound based scheme, a bisection clustered cell-free networking algorithm, $\text{B}\text{C}^2\text{F}$-Net, was further proposed.
Simulation results show that the proposed $\text{B}\text{C}^2\text{F}$-Net algorithm achieves similar sum ergodic capacity performance as the branch-and-bound based solution, yet with prominently lower computational complexity.
More importantly, it was shown that our $\text{B}\text{C}^2\text{F}$-Net algorithm outperforms the state-of-the-art benchmarks by achieving up to {25}\% higher average sum ergodic capacity.

Note that given the number of subnetworks, the clustered cell-free networking problem is reformulated as a convex integer programming problem in this paper and solved by the branch-and-bound method, which could be computationally costly.
In the future work, it is of great interest to explore relaxation methods such as semidefinite relaxation to further reduce the computational complexity of the clustered cell-free networking algorithm.
Furthermore, the path-loss coefficients between all UEs and all BSs are required to be collected at a central processor to run the proposed $\text{B}\text{C}^2\text{F}$-Net algorithm.
To avoid huge information exchange and reduce the channel measurement overhead, distributed clustered cell-free networking is desirable, which deserves much attention in future work.
Last but not least, while we focus on the uplink sum ergodic capacity in this paper, the downlink sum ergodic rate is also important, which further depends on the adopted precoding scheme and power allocation strategy.
Maximizing the downlink sum ergodic rate requires the joint optimization of clustered cell-free networking, power allocation and precoding design, which is another important topic to investigate in our future work.
\begin{figure*}[!b]
\hrule
\setcounter{equation}{49} 
\begin{align}
\label{appendix_b_2}
C_{sub}(\mathcal{C}_m \cup \mathcal{C}_n)&=\sum\limits_{b_l{\in}\mathcal{B}_{m}}\log_2\left(\frac{N_0+P\sum\limits_{u_k{\in}\mathcal{U}}q_{lk}^2}{N_0+P\sum\limits_{u_{k'}{\in}\mathcal{U}\setminus(\mathcal{U}_{m} \cup \mathcal{U}_{n})}q_{lk'}^2}\right){+}\sum\limits_{b_l{\in}\mathcal{B}_{n}}\log_2\left(\frac{N_0+P\sum\limits_{u_k{\in}\mathcal{U}}q_{lk}^2}{N_0+P\sum\limits_{u_{k'}{\in}\mathcal{U}\setminus(\mathcal{U}_{m} \cup \mathcal{U}_{n})}q_{lk'}^2}\right) \nonumber \\
&{\geq}\sum\limits_{b_l{\in}\mathcal{B}_{m}}\log_2\left(\frac{N_0+P\sum\limits_{u_k{\in}\mathcal{U}}q_{lk}^2}{N_0+P\sum\limits_{u_{k'}{\in}\mathcal{U}\setminus\mathcal{U}_{m}}q_{lk'}^2}\right){+}\sum\limits_{b_l{\in}\mathcal{B}_{n}}\log_2\left(\frac{N_0+P\sum\limits_{u_k{\in}\mathcal{U}}q_{lk}^2}{N_0+P\sum\limits_{u_{k'}{\in}\mathcal{U}\setminus\mathcal{U}_{n}}q_{lk'}^2}\right){=}C_{sub}(\mathcal{C}_m)+C_{sub}(\mathcal{C}_n) \tag{46}
\end{align}
\end{figure*}

{\appendices
\section*{APPENDIX A\\Proof of Theorem 1}
\label{Proof_of_NP-Hardness}
To prove that problem $\mathcal{P}1$ in (\ref{ul_capacity_maximization_problem_origin}) is NP-hard, let us first prove the NP-hardness of a simple case of problem $\mathcal{P}1$, where the maximum allowable per-subnetwork UE number $K_{max}$ is set to one.
In such a simple case, the number of subnetworks $M$ equals the number of UEs $K$ and the network reduces to a UE-centric network, where each UE $u_k$ forms its own subnetwork along with its serving BSs in $\mathcal{B}_k$, i.e., $\mathcal{C}_k=\{u_k\}\cup\mathcal{B}_k$.
As a result, the objective function of problem $\mathcal{P}1$ in (\ref{ul_capacity_maximization_problem_origin}) can be rewritten as
\begin{align}
\setcounter{equation}{42} 
\label{special_case_obj_func}
C_{sum}=\sum\limits_{u_k{\in}\mathcal{U}}\sum\limits_{b_l{\in}\mathcal{B}_k}\log_2\left(1+\frac{P{\cdot}q_{lk}^2}{N_0+P\sum\limits_{k'{\neq}k}q_{lk'}^2}\right),
\end{align}
by substituting (\ref{approximated_cluster_capacity_expr}) and (\ref{lambda_l}) into (\ref{ul_capacity_objctive_origin}).
To indicate which BSs are in which subnetworks, let us introduce an auxiliary variable $p_{lk}$.
If BS $b_l$ and UE $u_k$ are in the same subnetwork, $p_{lk}$ is set to $P$; otherwise, $p_{lk}$ is zero.
With the new auxiliary variable $p_{lk}$, the simple special case of problem $\mathcal{P}1$ is transformed to the following power control problem
\begin{subequations}
\begin{align}
\label{special_case_obj_func_final_form}
&\max\limits_{\{p_{lk}\}}\;C_{sum}{=}\sum\limits_{u_k{\in}\mathcal{U}}\sum\limits_{b_l{\in}\mathcal{B}_k}\hspace{-1mm}\log_2\hspace{-1mm}\left(\hspace{-1mm}1{+}\frac{p_{lk}}{\sigma_{lk}+\sum\limits_{k'{\neq}k}\alpha_{lk'}p_{lk}}\hspace{-1mm}\right) \\
&\text{s.t.}\;p_{lk}\in\{0,P\},{\;} l=1,2,\cdots,L;\; k=1,2,\cdots,K,
\end{align}
\end{subequations}
where $\sigma_{lk}{=}{N_0}/{q_{lk}^2}$ and $\alpha_{lk'}{=}{q_{lk'}^2}/{q_{lk}^2}$.
Such a power control problem for maximizing the sum ergodic capacity has been investigated in the literature and proved to be NP-hard \cite{dynamic_spectrum_anagement_complexity_duality}. 
Since the simple case of problem $\mathcal{P}1$ is NP-hard, and problem $\mathcal{P}1$ is harder than it due to the joint optimization of the number of subnetworks $M$ and the corresponding network decomposition $\mathcal{C}{=}\{\mathcal{C}_1,\mathcal{C}_2,\cdots, \mathcal{C}_M\}$, $\mathcal{P}1$ is also NP-hard.
The proof is completed.

\section*{APPENDIX B\\Proof of Theorem 2}
\label{Proof_of_Monotonicity}
Denote $\{\mathcal{C}^*_{1|M},\mathcal{C}^*_{2|M},\cdots,\mathcal{C}^*_{M|M}\}$ as the optimal network decomposition that maximizes the sum ergodic capacity when the number of subnetworks is fixed at $M$.
For a given number of subnetworks $M+1$, the maximum sum ergodic capacity can be written as
\begin{align}
\label{appendix_b_1}
&\sum\limits_{m=1}^{M+1}\hspace{-1mm}C_{sub}(\mathcal{C}^*_{m|M+1}) \nonumber \\
&{=}\hspace{-2.5mm}\sum\limits_{m=1}^{M-1}\hspace{-2mm}C_{sub}\hspace{-0.5mm}(\mathcal{C}^*_{m|M+1}\hspace{-0.5mm}){+}C_{sub}\hspace{-0.5mm}(\mathcal{C}^*_{M|M+1}){+}C_{sub}\hspace{-0.5mm}(\mathcal{C}^*_{M+1|M+1}\hspace{-0.5mm}).
\end{align}
Recall that the ergodic capacity of the $m$-th subnetwork $C_{sub}(\mathcal{C}_m)$ is given in (\ref{approximated_cluster_capacity_expr}) and (\ref{lambda_l}).
For two arbitrary subnetworks with $\mathcal{C}_m$, $\mathcal{C}_n \in \mathcal{V}$ and $\mathcal{C}_m \cap \mathcal{C}_n=\emptyset$, the sum ergodic capacity $C_{sub}(\mathcal{C}_m \cup \mathcal{C}_n)$ can be obtained from (\ref{approximated_cluster_capacity_expr}) and (\ref{lambda_l}) as (\ref{appendix_b_2}), which is shown at the bottom of this page.
By combining (\ref{ul_capacity_objctive_with_M*}) and (\ref{appendix_b_2}), the sum ergodic capacity $\sum\limits_{m=1}^{M+1}C_{sub}(\mathcal{C}^*_{m|M+1})$ is upper-bounded by
\begin{align}
\label{appendix_b_5}
&\sum\limits_{m=1}^{M+1}C_{sub}(\mathcal{C}^*_{m|M+1})\leq\sum\limits_{m=1}^{M-1}C_{sub}(\mathcal{C}^*_{m|M+1}) \nonumber \\
&+C_{sub}(\mathcal{C}_{M|M+1}^* \cup \mathcal{C}_{M+1|M+1}^*)\leq\sum\limits_{m=1}^{M}C_{sub}(\mathcal{C}^*_{m|M}), \tag{47}
\end{align}
because $\{\mathcal{C}^*_{1|M+1},$ $\mathcal{C}^*_{2|M+1},$ $\cdots,$ $\mathcal{C}^*_{M-1|M+1},$ $\mathcal{C}^*_{M|M+1}\cup\mathcal{C}^*_{M+1|M+1}\}$ can be considered as one possible network decomposition when the number of subnetworks is $M$ and the sum ergodic capacity with any network decomposition is no larger than that with the optimal network decomposition $\{\mathcal{C}_{1|M}^*,\mathcal{C}_{2|M}^*,\cdots,\mathcal{C}_{M|M}^*\}$ that maximizes the sum ergodic capacity.}
The proof is completed.

\bibliographystyle{IEEEtran}
\bibliography{IEEEabrv,reference}

\begin{thebibliography}{10}
\providecommand{\url}[1]{#1}
\csname url@samestyle\endcsname
\providecommand{\newblock}{\relax}
\providecommand{\bibinfo}[2]{#2}
\providecommand{\BIBentrySTDinterwordspacing}{\spaceskip=0pt\relax}
\providecommand{\BIBentryALTinterwordstretchfactor}{4}
\providecommand{\BIBentryALTinterwordspacing}{\spaceskip=\fontdimen2\font plus
\BIBentryALTinterwordstretchfactor\fontdimen3\font minus
  \fontdimen4\font\relax}
\providecommand{\BIBforeignlanguage}[2]{{%
\expandafter\ifx\csname l@#1\endcsname\relax
\typeout{** WARNING: IEEEtran.bst: No hyphenation pattern has been}%
\typeout{** loaded for the language `#1'. Using the pattern for}%
\typeout{** the default language instead.}%
\else
\language=\csname l@#1\endcsname
\fi
#2}}
\providecommand{\BIBdecl}{\relax}
\BIBdecl

\bibitem{xfn_ISIT2023}
F.~Xia and J.~Wang, ``Complexity-constrained clustered cell-free networking for
  sum capacity maximization,'' in \emph{Proc. IEEE Int. Symp. Inf. Theory
  (ISIT)}, Jun. 2023, pp. 2691--2696.

\bibitem{improving_dense_network_performance}
V.~Fernández-López, K.~I. Pedersen, B.~Soret, J.~Steiner, and P.~Mogensen,
  ``Improving dense network performance through centralized scheduling and
  interference coordination,'' \emph{IEEE Trans. Veh. Technol.}, vol.~66,
  no.~5, pp. 4371--4382, May 2017.

\bibitem{optimal_decomposition_networks}
L.~Dai and B.~Bai, ``Optimal decomposition for large-scale infrastructure-based
  wireless networks,'' \emph{IEEE Trans. Wireless Commun.}, vol.~16, no.~8, pp.
  4956--4969, Aug. 2017.

\bibitem{clustered_cell_free_networking}
J.~Wang, L.~Dai, L.~Yang, and B.~Bai, ``Clustered cell-free networking: A graph
  partitioning approach,'' \emph{IEEE Trans. Wireless Commun.}, vol.~22, no.~8,
  pp. 5349--5364, Aug. 2023.

\bibitem{rate_constrained_decomposition}
------, ``Rate-constrained network decomposition for clustered cell-free
  networking,'' in \emph{Proc. IEEE Int. Conf. Commun. (ICC)}, May 2022, pp.
  2549--2554.

\bibitem{c2_what_should_future_network_be}
L.~Yang, P.~Li, M.~Dong, B.~Bai, D.~Zaporozhets, X.~Chen, W.~Han, and B.~Li,
  ``C2: A capacity-centric architecture towards future wireless networking,''
  \emph{IEEE Trans. Wireless Commun.}, vol.~21, no.~10, pp. 8134--8147, Oct.
  2022.

\bibitem{cgn}
C.~Deng, L.~Yang, H.~Wu, D.~Zaporozhets, M.~Dong, and B.~Bai, ``{CGN}: A
  capacity-guaranteed network architecture for future ultra-dense wireless
  systems,'' in \emph{Proc. IEEE Int. Conf. Commun. (ICC)}, May 2022, pp.
  1853--1858.

\bibitem{zoy}
O.~Zhou, J.~Wang, and F.~Liu, ``Average downlink rate analysis for clustered
  cell-free networks with access point selection,'' in \emph{Proc. IEEE Int.
  Symp. Inf. Theory (ISIT)}, Jul. 2022, pp. 742--747.

\bibitem{das_JW}
J.~Wang and L.~Dai, ``Downlink rate analysis for virtual-cell based large-scale
  distributed antenna systems,'' \emph{IEEE Trans. Wireless Commun.}, vol.~15,
  no.~3, pp. 1998--2011, Mar. 2016.

\bibitem{ul_capacity_study_colocated_distributed_antennas}
L.~Dai, ``A comparative study on uplink sum capacity with co-located and
  distributed antennas,'' \emph{IEEE J. Sel. Areas Commun.}, vol.~29, no.~6,
  pp. 1200--1213, Jun. 2011.

\bibitem{asymptotic_rate_analysis_dl_das_JW}
J.~Wang and L.~Dai, ``Asymptotic rate analysis of downlink multi-user systems
  with co-located and distributed antennas,'' \emph{IEEE Trans. Wireless
  Commun.}, vol.~14, no.~6, pp. 3046--3058, Jun. 2015.

\bibitem{distributed_vs_microcellular}
H.~Zhu, ``Performance comparison between distributed antenna and microcellular
  systems,'' \emph{IEEE J. Sel. Areas Commun.}, vol.~29, no.~6, pp. 1151--1163,
  Jun. 2011.

\bibitem{network_coordination_comp}
M.~Karakayali, G.~Foschini, and R.~Valenzuela, ``Network coordination for
  spectrally efficient communications in cellular systems,'' \emph{IEEE
  Wireless Commun.}, vol.~13, no.~4, pp. 56--61, Aug. 2006.

\bibitem{comp_uc_adaptive_clustering}
V.~Garcia, Y.~Zhou, and J.~Shi, ``Coordinated multipoint transmission in dense
  cellular networks with user-centric adaptive clustering,'' \emph{IEEE Trans.
  Wireless Commun.}, vol.~13, no.~8, pp. 4297--4308, Aug. 2014.

\bibitem{network_mimo_origin}
D.~Gesbert, S.~Hanly, H.~Huang, S.~Shamai~Shitz, O.~Simeone, and W.~Yu,
  ``Multi-cell {MIMO} cooperative networks: A new look at interference,''
  \emph{IEEE J. Sel. Areas Commun.}, vol.~28, no.~9, pp. 1380--1408, Dec. 2010.

\bibitem{dynamic_coalition_network_mimo}
S.~Guruacharya, D.~Niyato, M.~Bennis, and D.~I. Kim, ``Dynamic coalition
  formation for network {MIMO} in small cell networks,'' \emph{IEEE Trans.
  Wireless Commun.}, vol.~12, no.~10, pp. 5360--5372, Oct. 2013.

\bibitem{c_ran_toward_green_ran}
X.~Wang, Y.~Huang, C.~Cui, K.~Chen, and M.~Chen, ``{C-RAN}: Evolution toward
  green radio access network,'' \emph{China Commun.}, vol.~7, pp. 107--112,
  Jul. 2010.

\bibitem{c_ran_overview}
A.~{Checko}, H.~L. {Christiansen}, Y.~{Yan}, L.~{Scolari}, G.~{Kardaras}, M.~S.
  {Berger}, and L.~{Dittmann}, ``Cloud {RAN} for mobile networks -- {A}
  technology overview,'' \emph{IEEE Commun. Surv. Tutor.}, vol.~17, no.~1, pp.
  405--426, Q1 2015.

\bibitem{pa_robust_tx_uc_c_ran}
C.~Pan, H.~Mehrpouyan, Y.~Liu, M.~Elkashlan, and N.~Arumugam, ``Joint pilot
  allocation and robust transmission design for ultra-dense user-centric {TDD
  C-RAN} with imperfect {CSI},'' \emph{IEEE Trans. Wireless Commun.}, vol.~17,
  no.~3, pp. 2038--2053, Mar. 2018.

\bibitem{cf_vs_small_cell}
H.~Q. {Ngo}, A.~{Ashikhmin}, H.~{Yang}, E.~G. {Larsson}, and T.~L. {Marzetta},
  ``Cell-free massive {MIMO} versus small cells,'' \emph{IEEE Trans. Wireless
  Commun.}, vol.~16, no.~3, pp. 1834--1850, Mar. 2017.

\bibitem{user_centric_cf_mmimo_survey}
H.~A. Ammar, R.~Adve, S.~Shahbazpanahi, G.~Boudreau, and K.~V. Srinivas,
  ``User-centric cell-free massive {MIMO} networks: A survey of opportunities,
  challenges and solutions,'' \emph{IEEE Commun. Surv. Tutor.}, vol.~24, no.~1,
  pp. 611--652, Jan. 2022.

\bibitem{joint_power_antenna_selection_c_ran}
A.~Liu and V.~K.~N. Lau, ``Joint power and antenna selection optimization in
  large cloud radio access networks,'' \emph{IEEE Trans. Signal Process.},
  vol.~62, no.~5, pp. 1319--1328, Mar. 2014.

\bibitem{stable_matching_urllc}
T.~Hößler, P.~Schulz, E.~A. Jorswieck, M.~Simsek, and G.~P. Fettweis,
  ``Stable matching for wireless {URLLC} in multi-cellular, multi-user
  systems,'' \emph{IEEE Trans. Commun.}, vol.~68, no.~8, pp. 5228--5241, Aug.
  2020.

\bibitem{ee_ua_pa_multi_conn_mmwave_net}
K.~Jin, X.~Cai, J.~Du, H.~Park, and Z.~Tang, ``Toward energy efficient and
  balanced user associations and power allocations in multiconnectivity-enabled
  mmwave networks,'' \emph{IEEE Trans. Green Commun. Netw.}, vol.~6, no.~4, pp.
  1917--1931, Dec. 2022.

\bibitem{multi_conn_ua_pa_mmwave_net}
X.~Cai, A.~Chen, L.~Chen, and Z.~Tang, ``Joint optimal multi-connectivity
  enabled user association and power allocation in mmwave networks,'' in
  \emph{Proc. IEEE Wireless Commun. Netw. Conf. (WCNC)}, 2021, pp. 1--6.

\bibitem{ee_ua_iiot}
X.~Jian, L.~Wu, K.~Yu, M.~Aloqaily, and J.~Ben-Othman, ``Energy-efficient user
  association with load-balancing for cooperative {IIoT} network within {B5G}
  era,'' \emph{J. Netw. Comput. Appl.}, vol. 189, Sep. 2021.

\bibitem{ul_capacity_analysis_das}
L.~Dai, ``An uplink capacity analysis of the distributed antenna system
  ({DAS}): From cellular {DAS} to {DAS} with virtual cells,'' \emph{IEEE Trans.
  Wireless Commun.}, vol.~13, no.~5, pp. 2717--2731, May 2014.

\bibitem{uc_jt_virtual_cell_udn}
Y.~Zhang, S.~Bi, and Y.-J.~A. Zhang, ``User-centric joint transmission in
  virtual-cell-based ultra-dense networks,'' \emph{IEEE Trans. Veh. Technol.},
  vol.~67, no.~5, pp. 4640--4644, May 2018.

\bibitem{uc_5g_cellular_network}
S.~Buzzi, C.~D’Andrea, A.~Zappone, and C.~D’Elia, ``User-centric 5{G}
  cellular networks: Resource allocation and comparison with the cell-free
  massive {MIMO} approach,'' \emph{IEEE Trans. Wireless Commun.}, vol.~19,
  no.~2, pp. 1250--1264, Feb. 2020.

\bibitem{min_separation_clustering_udn}
Y.~Xiao, X.~Wu, W.~Wu, and S.~Han, ``Minimum separation clustering algorithm
  with high separation degree in ultra-dense network,'' in \emph{Proc. IEEE
  Global Commun. Conf. (GLOBECOM)}, Dec. 2019, pp. 1--6.

\bibitem{rance_clustering_adhoc}
X.~Chen, G.~Sun, T.~Wu, L.~Liu, H.~Yu, and M.~Guizani, ``{RANCE}: {A} randomly
  centralized and on-demand clustering protocol for mobile ad hoc networks,''
  \emph{IEEE Internet Things J.}, vol.~9, no.~23, pp. 23\,639--23\,658, Dec.
  2022.

\bibitem{thwsn_clustering_sensor_network}
N.~Kumar, P.~Rani, V.~Kumar, S.~V. Athawale, and D.~Koundal, ``{THWSN}:
  {Enhanced} energy-efficient clustering approach for three-tier heterogeneous
  wireless sensor networks,'' \emph{IEEE Sens. J.}, vol.~22, no.~20, pp.
  20\,053--20\,062, Oct. 2022.

\bibitem{clustered_d2d_net}
K.~S. Khan, A.~Naeem, and A.~Jamalipour, ``Incentive-based caching and
  communication in a clustered {D2D} network,'' \emph{IEEE Internet Things J.},
  vol.~9, no.~5, pp. 3313--3320, Mar. 2022.

\bibitem{static_clustering_comp}
P.~Marsch and G.~Fettweis, ``Static clustering for cooperative multi-point
  {(CoMP)} in mobile communications,'' in \emph{Proc. IEEE Int. Conf. Commun.
  (ICC)}, Jun. 2011.

\bibitem{dynamic_cell_clustering_design}
H.~Sun, X.~Zhang, and W.~Fang, ``Dynamic cell clustering design for realistic
  coordinated multipoint downlink transmission,'' in \emph{Proc. IEEE Int.
  Symp. Pers. Indoor Mobile Radio Commun. (PIMRC)}, Sep. 2011, pp. 1331--1335.

\bibitem{dynamic_rr_clustering_user_scheduling}
Y.~Du and G.~de~Veciana, ``Wireless networks without edges: Dynamic radio
  resource clustering and user scheduling,'' in \emph{Proc. IEEE Conf. Comput.
  Commun. (INFOCOM)}, Apr. 2014, pp. 1321--1329.

\bibitem{dynamic_clustering_multicell_cooperative_processing}
A.~Papadogiannis, D.~Gesbert, and E.~Hardouin, ``A dynamic clustering approach
  in wireless networks with multi-cell cooperative processing,'' in \emph{Proc.
  IEEE Int. Conf. Commun. (ICC)}, May 2008, pp. 4033--4037.

\bibitem{dynamic_clustering_mu_das}
J.~Liu and D.~Wang, ``An improved dynamic clustering algorithm for multi-user
  distributed antenna system,'' in \emph{Proc. IEEE Int. Conf. Wireless Commun.
  Signal Process. (WCSP)}, Nov. 2009.

\bibitem{virtual_cell_clustering_ra}
M.~Yemini and A.~J. Goldsmith, ``Virtual cell clustering with optimal resource
  allocation to maximize capacity,'' \emph{IEEE Trans. Wireless Commun.},
  vol.~20, no.~8, pp. 5099--5114, Aug. 2021.

\bibitem{ee_large_scale_das}
J.~Joung, Y.~K. Chia, and S.~Sun, ``Energy-efficient, large-scale
  distributed-antenna system ({L-DAS}) for multiple users,'' \emph{IEEE J. Sel.
  Top. Signal Process.}, vol.~8, no.~5, pp. 954--965, Oct. 2014.

\bibitem{cf_mmimo_joint_uc_aps}
R.~Wang, M.~Shen, Y.~He, and X.~Liu, ``Performance of cell-free massive {MIMO}
  with joint user clustering and access point selection,'' \emph{IEEE Access},
  vol.~9, pp. 40\,860--40\,870, Feb. 2021.

\bibitem{clustered_cf_mmimo}
F.~Riera-Palou, G.~Femenias, A.~G. Armada, and A.~Pérez-Neira, ``Clustered
  cell-free massive {MIMO},'' in \emph{IEEE Globecom Workshops}, Dec. 2018, pp.
  1--6.

\bibitem{tse_fundamentals_wireless_comm}
D.~Tse and V.~Pramod, \emph{\BIBforeignlanguage{English (US)}{{Fundamentals of
  Wireless Communication}}}.\hskip 1em plus 0.5em minus 0.4em\relax Cambridge
  University Press, Jan. 2005.

\bibitem{static_clustering_cran}
N.~Wang, T.~Wang, X.~Han, X.~Xu, and G.~Li, ``Group decoding with static
  clustering for cloud radio access networks,'' in \emph{Proc. IEEE Int. Conf.
  Commun. Technol. (ICCT)}, Oct. 2021, pp. 182--186.

\bibitem{dl_ra_cf_mmimo_uc_clustering}
H.~A. Ammar, R.~Adve, S.~Shahbazpanahi, G.~Boudreau, and K.~V. Srinivas,
  ``Downlink resource allocation in multiuser cell-free {MIMO} networks with
  user-centric clustering,'' \emph{IEEE Trans. Wireless Commun.}, vol.~21,
  no.~3, pp. 1482--1497, Mar. 2022.

\bibitem{greedy_mu_scheduling_clustered_cf_mmimo}
S.~Mashdour, R.~C. de~Lamare, and J.~P. S.~H. Lima, ``Enhanced subset greedy
  multiuser scheduling in clustered cell-free massive {MIMO} systems,''
  \emph{IEEE Commun. Lett.}, vol.~27, no.~2, pp. 610--614, Feb. 2023.

\bibitem{algebraic_graph_theory_book}
C.~Godsil and G.~F. Royle, \emph{Algebraic Graph Theory}.\hskip 1em plus 0.5em
  minus 0.4em\relax Springer Science \& Business Media, 2013.

\bibitem{branch_and_bound}
P.~Bonami, L.~T. Biegler, A.~R. Conn, G.~Cornuéjols, I.~E. Grossmann, C.~D.
  Laird, J.~Lee, A.~Lodi, F.~Margot, N.~Sawaya, and A.~Wächter, ``An
  algorithmic framework for convex mixed integer nonlinear programs,''
  \emph{Discrete Optimization}, vol.~5, no.~2, pp. 186--204, 2008.

\bibitem{mosek}
D.~MOSEK Inc.~Copenhagen, ``{MOSEK ApS},'' [Online]. Available:
  \url{www.mosek.com}, accessed 15.04.2022.

\bibitem{lectures_modern_convex_opt}
A.~Ben-Tal and A.~Nemirovski, ``Lectures on modern convex optimization:
  analysis, algorithms, and engineering applications,'' in \emph{MPSSIAM Series
  on Optimization}.\hskip 1em plus 0.5em minus 0.4em\relax Philadelphia, PA,
  USA: SIAM, 2001.

\bibitem{1_bit_mmimo_precoding_partial_bb}
A.~Li, F.~Liu, C.~Masouros, Y.~Li, and B.~Vucetic, ``Interference exploitation
  1-bit massive {MIMO} precoding: A partial branch-and-bound solution with
  near-optimal performance,'' \emph{IEEE Trans. Wireless Commun.}, vol.~19,
  no.~5, pp. 3474--3489, May 2020.

\bibitem{dynamic_spectrum_anagement_complexity_duality}
Z.-Q. Luo and S.~Zhang, ``Dynamic spectrum management: Complexity and
  duality,'' \emph{IEEE J. Sel. Top. Signal Process.}, vol.~2, no.~1, pp.
  57--73, Feb. 2008.

\end{thebibliography}
\begin{IEEEbiography}[{\includegraphics[width=1in,height=1.25in,clip,keepaspectratio]{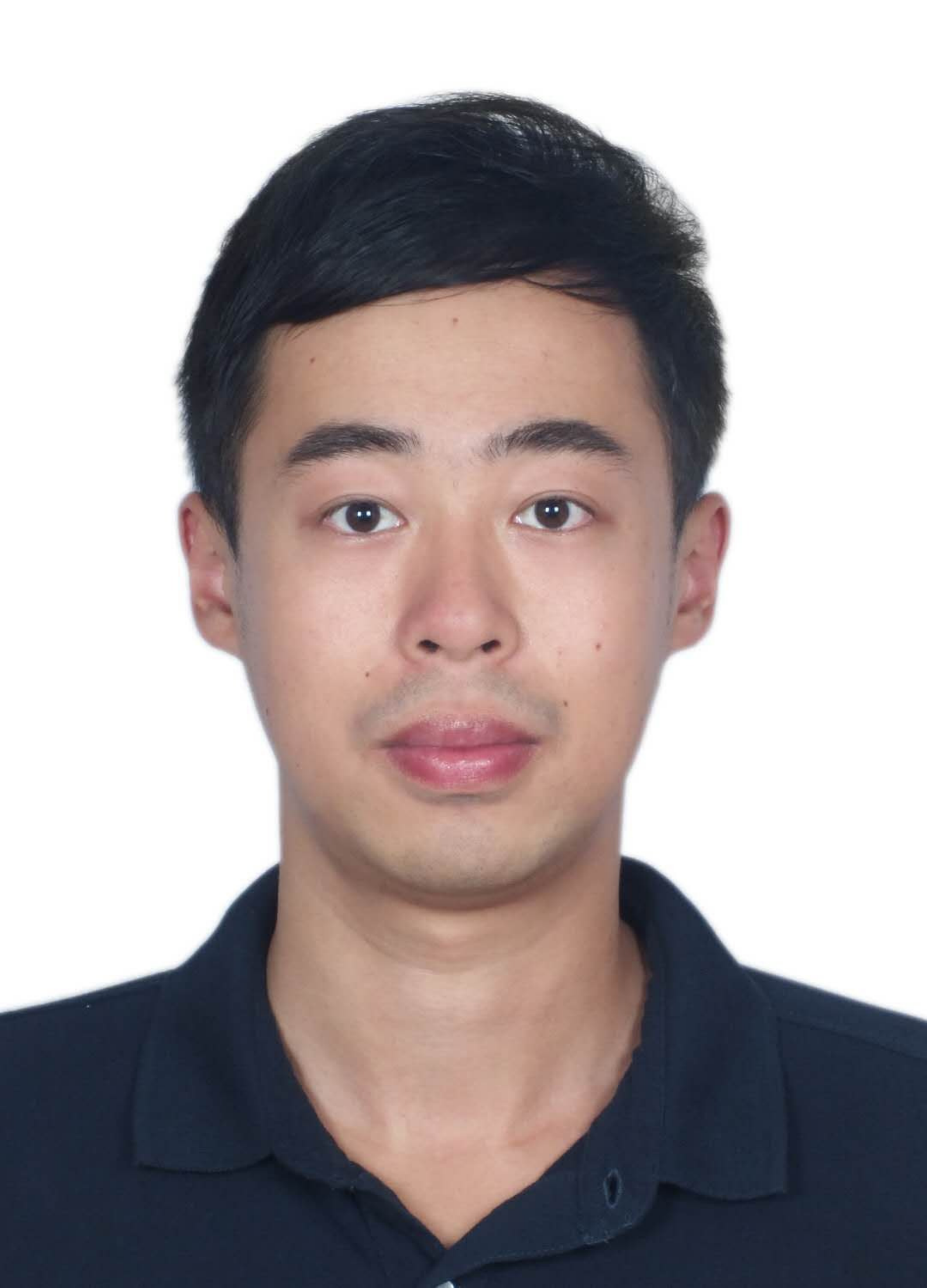}}]{Funing Xia}
Funing Xia received the B.S. degree in electronic engineering from Xidian University, Xi’an, China, in 2016, and the Diploma (Dipl.-Ing.) degree in electronic engineering from Technische Universität Dresden, Germany, in 2020. He is currently pursuing the Ph.D. degree with the College of Electronic and Information Engineering, Tongji University, Shanghai, China. His research interests include clustered cell-free networking, energy efficiency maximization, and transmission optimization for cell-free networks.
\end{IEEEbiography}

\begin{IEEEbiography}[{\includegraphics[width=1in,height=1.25in,clip,keepaspectratio]{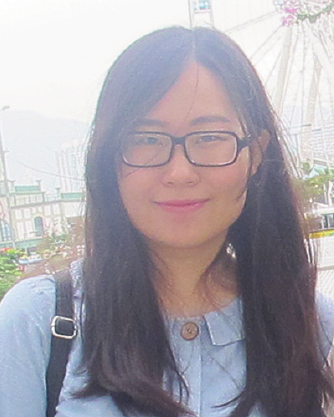}}]{Junyuan Wang}
(Member, IEEE) received the B.S. degree in communications engineering from Xidian University, Xi’an, China, in 2010, and the Ph.D. degree in electronic engineering from City University of Hong Kong, Hong Kong, China, in 2015. From 2015 to 2017, she was a Research Associate at the School of Engineering and Digital Arts, University of Kent, Canterbury, U.K. From 2018 to 2020, she was a Lecturer (an Assistant Professor) at the Department of Computer Science, Edge Hill University, Ormskirk, U.K. She is currently a Research Professor with the College of Electronic and Information Engineering and the Institute for Advanced Study, Tongji University, Shanghai, China. Her research interests include wireless communications and networking, and artificial intelligence. She was a co-recipient of the Best Paper Award from the IEEE International Conference on Communications in China (ICCC) in 2024, a co-recipient of the Best Student Paper Award from the IEEE 85th Vehicular Technology Conference–Spring (VTC-Spring) in 2017, and a recipient of the Shanghai Leading Talent Program (Young Scientist) in 2021.
\end{IEEEbiography}

\begin{IEEEbiography}[{\includegraphics[width=1in,height=1.25in,clip,keepaspectratio]{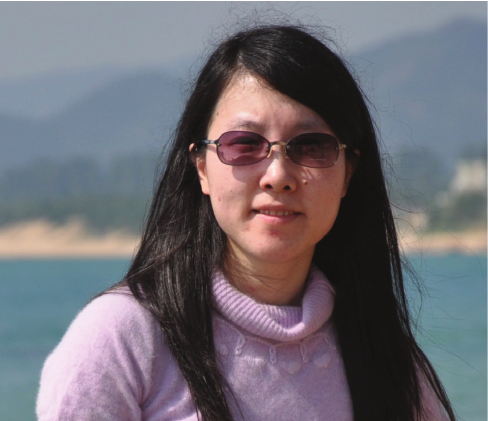}}]{Lin Dai}
(Senior Member, IEEE) received the B.S. degree in electronic engineering from the Huazhong University of Science and Technology, Wuhan, China, in 1998, and the M.S. and Ph.D. degrees in electronic engineering from Tsinghua University, Beijing, China, in 2003. She was a Post-Doctoral Fellow at The Hong Kong University of Science and Technology and the University of Delaware. Since 2007, she has been with City University of Hong Kong, where she is currently a Full Professor. She has broad interests in communications and networking theory, with special interest in wireless communications. She was a co-recipient of the Best Paper Award from the IEEE Wireless Communications and Networking Conference (WCNC) in 2007 and the IEEE Marconi Prize Paper Award in 2009.
\end{IEEEbiography}

\end{document}